\def\draft{0}

\documentclass{llncs}
\usepackage{amsfonts, amsmath, amssymb}
\usepackage{url, graphics}
\usepackage[dvips]{graphicx}
\usepackage{color}
\usepackage{fullpage}

\def\FullBox{\hbox{\vrule width 8pt height 8pt depth 0pt}}

\def\qed{\ifmmode\qquad\FullBox\else{\unskip\nobreak\hfil
\penalty50\hskip1em\null\nobreak\hfil\FullBox
\parfillskip=0pt\finalhyphendemerits=0\endgraf}\fi}

\def\qedsketch{\ifmmode\Box\else{\unskip\nobreak\hfil
\penalty50\hskip1em\null\nobreak\hfil$\Box$
\parfillskip=0pt\finalhyphendemerits=0\endgraf}\fi}

\newcommand{\ie} {{\it i.e.\ }}

\newcommand{\zo}{\{0,1\}}

\newcommand{\E}{\mathop{\mathbb E}\displaylimits}

\newcommand{\eps}{\varepsilon}

\ifnum\draft=1
\newcommand{\authnote}[2]{{ \bf [#1's Note: #2]}}
\else
\newcommand{\authnote}[2]{}
\fi

\newcommand{\COMMENT}[1]{}
\newcommand{\ket}[1]{|#1\rangle}
\newcommand{\bra}[1]{\langle#1|}
\newcommand{\ketbra}[2]{|#1\rangle\langle#2|}
\def\01{\{0,1\}}
\newcommand{\braket}[2]{\langle{#1}|{#2}\rangle} 
\def\01{\{0,1\}}

\newcommand{\triple}[3]{\langle{#1}|{#2}|{#3}\rangle}

\newcommand{\norm}[1]{\mbox{$\parallel{#1}\parallel$}}
\newcommand{\Tr}{\mbox{\rm Tr}}

\newcommand{\cadre}[1]
{
\begin{tabular}{|p{0.9\textwidth}|}
\hline
#1 \\
\hline
\end{tabular}
}

\newcommand{\spa}[1]{\mathcal{#1}}

\newcommand{\altketbra}[1]{\ketbra{#1}{#1}}

\newcommand{\tnorm}[1]{\norm{#1}_{\mathrm{tr}}}
\newcommand{\F}{\operatorname{F}}

\title{Parallel Repetition of Entangled Games with Exponential Decay via the Superposed Information Cost}
\author{Andr\'e Chailloux\inst{1} \and Giannicola Scarpa\inst{2}}
\institute{SECRET Project Team, INRIA Paris-Rocquencourt \and
CWI, Amsterdam}
\date{}

\begin{document}
\maketitle

\begin{abstract}
In a two-player game, two cooperating but non communicating players, Alice and Bob, receive inputs taken from a probability distribution. Each of them produces an output and they win the game if they satisfy some predicate on their inputs/outputs. The entangled value $\omega^*(G)$ of a game $G$ is the maximum probability that Alice and Bob can win the game if they are allowed to share an entangled state prior to receiving their inputs. 

The $n$-fold parallel repetition $G^n$ of $G$ consists of $n$ instances of $G$ where the players receive all the inputs at the same time and produce all the outputs at the same time. They win $G^n$ if they win each instance of $G$. 

In this paper we show that for any game $G$ such that $\omega^*(G) = 1 - \eps < 1$, $\omega^*(G^n)$ decreases exponentially in $n$. First, for any game $G$ on the uniform distribution, we show that $\omega^*(G^n) = (1 - \eps^2)^{\Omega\left(\frac{n}{\log(|I||O|)} - |\log(\eps)|\right)}$, where $|I|$ and $|O|$ are the sizes of the input and output sets. From this result, we  show that for any entangled game $G$, $\omega^*(G^n) \le (1 - \eps^2)^{\Omega(\frac{n}{Q\log(|I||O|)} - \frac{|\log(\eps)|}{Q})}$ where $p$ is the input distribution of $G$ and $Q= \frac{|I|^2 \max_{xy} p_{xy}^2 }{\min_{xy} p_{xy} }$. This implies parallel repetition with exponential decay as long as $\min_{xy} \{p_{xy}\} \neq 0$ for general games.
To prove this parallel repetition, we introduce the concept of \emph{Superposed Information Cost} for entangled games which is inspired from the information cost used in communication complexity.
\end{abstract}

\section{Introduction}
A \emph{two-player (nonlocal) game} is played between two cooperating parties Alice and Bob which are not allowed to communicate.  This game $G$ is characterized by an input set $I$, an output set $O$, a probability distribution $p$ on $I^2$
and a result function $V: O^2 \times I^2 \rightarrow \zo$. The game proceeds as follows: Alice receives $x \in I$, Bob receives $y \in I$ where $(x,y)$ is taken according to $p$. Alice outputs $a \in O$ and Bob outputs $b \in O$. They win the game if $V(a,b|x,y) = 1$. The value of the game $\omega(G)$ is the maximum probability, over all strategies, with which Alice and Bob can win the game. 

The $n$-fold parallel repetition $G^n$ of $G$ consists of the following. Alice and Bob get inputs $x_1,\dots,x_n$ and $y_1,\dots,y_n$, respectively. Each $(x_i,y_i)$ is taken independently according to $p$. They output $a_1,\dots,a_n$ and $b_1,\dots,b_n$, respectively. They win the game if and only if $\forall i, \ V(a_i,b_i|x_i,y_i) = 1$. In order to win the $n$-fold repetition, Alice and Bob can just take the best strategy for $G$ and use it $n$ times. If they do so, they will win $G^n$ with probability $(\omega(G))^n$ which shows that $\omega(G^n) \geq (\omega(G))^n$.

Parallel repetition of games studies how the quantity $\omega(G^n)$ behaves. For example, if $\omega(G^n) = (\omega(G))^n$ for each $n$ then we say that $G$ admits perfect parallel repetition. However, there are some games for which this does not hold, for example the CHSH game \cite{CHSH69} repeated two times. It was a long-standing open question to determine whether the value of $\omega(G^n)$ decreases exponentially in $n$. This was first shown by Raz~\cite{Raz98}. Afterwards, a series of works showed improved results for specific types of games~\cite{Hol07,Rao08,AKK+08}. Parallel repetition for games has many applications, from direct product theorems in communication complexity~\cite{PRW97} to hardness of approximation results~\cite{BGS98,Fei98,Has01}.

In the quantum setting, it is natural to consider games where Alice and Bob are allowed to share some entangled state at the beginning of the game. 
In this case we talk about entangled strategies. The maximum probability that Alice and Bob can win a game $G$, over all the entangled strategies, is the entangled value $\omega^*(G)$.
Some entangled games are witnesses for the phenomenon of quantum non-locality, as they are special cases of the so-called Bell inequality violations. (We have a Bell inequality violation whenever $\omega^*(G) > \omega(G)$.)
The study of entangled games is also greatly related to our understanding of quantum entanglement.

Perfect parallel repetition has been shown for entangled XOR games~\cite{CSUU08}. It was also shown that entangled unique games~\cite{KRT08} admit parallel repetition with exponential decay. Finally, it was shown that any entangled game admits (a variant of) parallel repetition~\cite{KV11}. However, this last parallel repetition only shows a polynomial decay of $\omega^*(G^n)$. It was unknown for a large class of games whether this decay is exponential or not. Very recently two more works have been presented: a parallel repetition result with exponential decay for entangled projection games~\cite{DSV13} and an independent work~\cite{JPY13} similar to this one.

\setcounter{footnote}{0}

\subsection{Contribution}
The main contribution of this paper is the following theorem. 

\begin{theorem}\label{Theorem:Final}
For any game $G$ on the uniform distribution with $\omega^*(G) \le 1 - \eps$, we have: $$\omega^*(G^n) = (1 - \eps^2)^{\Omega\left(\frac{n}{\log(|I||O|)} - |\log(\eps)|\right)}.$$
where $|I|$ and $|O|$ are respectively the size of the input and the output sets.
\end{theorem}

The class of entangled games with a uniform distribution is a large class of entangled games for which such parallel repetition was unknown. 
We can extend this result to any entangled game.

\begin{corollary}\label{corollary:Final}
For any game $G$ such that $\omega^*(G) \le 1 - \eps$, we have that $$\omega^*(G^n) \le (1 - \eps^2)^{\Omega(\frac{n}{Q\log(|I||O|)} - \frac{|\log(\eps)|}{Q})},$$
where $I$ and $O$ are respectively the input and output sets and $Q= \frac{k^2 \max_{xy} p_{xy}^2 }{\min_{xy} p_{xy} }$.
\end{corollary}

This corollary can be obtained directly from the previous theorem. It is not as strong as usual parallel repetition theorems with exponential decay because of this dependency on $Q$. Notice however that $Q$ depends only on the game $G$ and not on $n$.

\textbf{Remark:} In a previous version of this paper, we had a different claim which had a flaw in the proof. We replaced it by the above Corollary which is weaker in the sense that it gives non trivial bounds only for the case where $\min_{xy} \{p_{xy}\} \neq 0$.
 
\subsection{Superposed Information cost}
In order to prove the main theorem, we introduce the concept of \emph{Superposed Information Cost} of a game, an insightful concept and the cornerstone of our proof.

This concept is derived from the notion of information cost widely used in communication complexity \cite{CSWY01,BJKS04,Bra12,KLL+12}. In the setting of communication complexity, we consider a function $f(x,y)$ and suppose that Alice has some input $x$ and
Bob some input $y$. They want to determine the outcome of $f(x,y)$ for a certain function $f$ with the minimal amount of communication. The interactive information cost $IC$ of $f$ describes the least amount of information that Alice and Bob need  to have about each other's inputs in order to compute $f(x,y)$.

We want to follow a similar approach for entangled games.
In entangled games, the quantum state Alice and Bob share is independent of the inputs $x,y$. 
We now give extra resources to Alice and Bob: advice states. Alice and Bob are given an advice state $\ket{\phi_{xy}}$ that can depend on their inputs. This can greatly increase their winning probability. For example, Alice could have perfect knowledge of Bob's input $y$, and vice-versa. 

We define (informally) the information cost of a game as follows:
\\

\cadre
{
\textbf{Information Cost for entangled games}

Alice and Bob are given advice states $\ket{\phi_{xy}}$ to share that can depend on their inputs.
What is the minimal amount of information that these states have to give Alice and Bob about each other's input, in order to allow them to win the game with probability $1$?
} $ \ $ \\

This is a natural extension of the information cost to entangled games. However, it is a limited notion since we cannot relate it to the entangled value of the game. (A simple counterexample can be obtained from the CHSH game.) Therefore, we extended this notion to the case where we allow the players to be \emph{in a superposition of their inputs.} 
 \\

\cadre
{
\textbf{Superposed Information Cost (SIC) for entangled games}

We extend the notion of information cost by allowing the players to have a superposition of their inputs.
We then consider the amount of information that advice states have to give Alice and Bob about each other's input, in order to allow them to win with probability $1$.
 } $ \ $  \\

These notions are defined precisely in Section \ref{Section:SuperposedDefs}.


\subsubsection{Lower bounding the value of entangled games using the superposed information cost.}
The reason we introduce the superposed information cost for entangled games is that we want to have an information theoretic characterization of the value of entangled games. The next theorem states that the value of any entangled game on the uniform distribution can be lower bounded by the superposed information cost (this does not hold for the non-superposed one).

\begin{theorem}\label{Theorem:OneShotPerfect}
For any game $G$ with a uniform input distribution, we have  
$SIC(G) \ge \frac{1 - \omega^*(G)}{32\ln(2)}$ or equivalently $\omega^*(G) \ge 1 - 32\ln(2) \cdot SIC(G)$.
\end{theorem}


The Superposed information cost is additive under parallel repetition:
\begin{proposition}\label{Prop2}
$SIC(G^n) = n SIC(G)$.
\end{proposition}

Putting these two results together, we have $SIC(G^n) \ge \frac{n (1-\omega^*(G))}{32\ln(2)}$. This result shows that $SIC(G^n)$ is large when $n$ increases and can be seen as evidence that the game $G^n$ is hard to win and that $\omega^*(G^n)$ decreases fast. 

\subsubsection{Using SIC to show our parallel repetition theorem.}
We fix a game $G$ with $\omega^*(G) = 1 - \eps$ and $\omega^*(G^n) = 2^{-t}$ for some $t$. In order to prove our theorem, we consider a quantity $S$ which is strongly related to $SIC(G^n)$. We show that 
\begin{equation} \label{eq:S_quantity}
\Omega({n \eps}) \le S \le O\left(\frac{t \log(|I||O|)}{\eps}\right).  
\end{equation}
The lower bound is a natural extension of the above argument about the additivity of SIC. The ingredient we need to show the upper bound is the following \emph{communication task}:
\begin{itemize}
\item The players use an optimal strategy for $G^n$ and win with probability\mbox{ $\omega^*(G^n) = 2^{-t}$.}
\item Alice sends $m = O(\frac{t \log(|I||O|)}{\eps})$ bits to Bob.
\item Using this message, Bob's goal is to determine with high probability whether they won most of the games or not.
\end{itemize}
Switching to a communication task and to a related quantity $S$ seems much weaker than showing directly an upper bound on $SIC(G^n)$, but it will be enough for us. Combining these two results, we conclude that $t = \Omega(\frac{n \eps^2}{\log(|I||O|)})$ or equivalently, for $\eps$ close to $0$, $\omega^*(G^n) = (1 - \eps^2)^{\Omega(\frac{n}{\log(|I| |O|)})}$.

\subsection{Organization of the paper}
Section~\ref{Section:EntangledGames} contains preliminaries about entangled games. In Section~\ref{Section:FirstDefs}, we define the key concept of the superposed information cost for a game and show that this quantity is additive when repeating games in parallel. In Section~\ref{Section:Orga}, we provide a brief organization of the main proof. In Section~\ref{Section:Single}, we show Theorem~\ref{Theorem:OneShotPerfect} and some generalizations. In Section~\ref{Sec:UpperBound} we derive the upper bound of~\eqref{eq:S_quantity} (the lower bound is proven in the main paper). Finally, in Section \ref{Section:finalThm} we prove our main theorem. Many proofs are deferred to the Appendix.


\section{Entangled Games}\label{Section:EntangledGames}
\subsubsection{The value of an entangled game} 
\begin{definition}
An entangled game $G = (I,O,V,p)$ is defined by finite input and output sets $I$ and $O$ as well as an accepting function $V: O^2 \times I^2 \rightarrow \zo$ and a probability distribution $p: I^2 \rightarrow [0,1]$.
\end{definition}

A strategy for the game proceeds as follows. Alice and Bob can share any quantum state. Then, Alice receives an input $x \in I$ and Bob receives an input $y \in I$ where these inputs are sampled according to $p$. They can perform any quantum operation but are not allowed to communicate. Alice outputs $a \in O$ and Bob outputs $b \in O$. They win the game if $V(a,b|x,y) = 1$. 

The \emph{entangled value} of a game $G$ is the maximal probability with which Alice and Bob can win the game. From standard purification techniques, we have that w.l.o.g., Alice and Bob share a pure state $\ket{\phi}$ and their optimal strategy consists of projective measurements $A^x = \{A^x_a\}_{a \in O}$ and $B^y = \{B^y_b\}_{b \in O}$ on $\ket{\phi}$.
This means that after receiving their inputs, they share a state of the form 
$
\rho = \sum_{x,y \in I} p_{xy} \altketbra{x} \otimes \altketbra{\phi} \otimes \altketbra{y}, 
$
for some state $\ket{\phi}$. 
\begin{definition}
The entangled value of a game $G$ is 
$$
\omega^*(G) = \sup_{\ket{\phi},A^x,B^y} \sum_{x,y,a,b} p_{xy} V(a,b|x,y) \triple{\phi}{A^x_a \otimes B^y_b}{\phi}.
$$
\end{definition}

\begin{definition}
A game $G = (I,O,V,p)$ is on the uniform distribution if $I = [k]$ for some $k$ and $\forall x,y \in [k], \ p_{xy} = \frac{1}{k^2}$. We write $p = \mbox{Unif.}$ when this is the case.
\end{definition}
\subsubsection{Value of a game with advice states}
\label{Section:advice}
Consider a game $G = (I,O,V,p)$. We are interested in the value of the game when the two players share an advice state $\ket{\phi_{xy}}$ additionally to their inputs $x,y$. This means that Alice and Bob share a state of the form 
$
\rho = \sum_{x,y,a,b} p_{xy} \altketbra{x} \otimes \altketbra{\phi_{xy}} \otimes \altketbra{y}.
$

\begin{definition}
The entangled value of  $G$, given that Alice and Bob share the above state $\rho$ is
$$
\omega^*(G|\rho) = \max_{A^x,B^y} \sum_{x,y} p_{xy} V(a,b|x,y) \triple{\phi_{xy}}{A^x_a \otimes B^y_b}{\phi_{xy}}.
$$
\end{definition}
\subsubsection{Repetition of entangled games}

In the $n$-fold parallel repetition of a game $G$, each player gets $n$ inputs from $I$ and must produce $n$ outputs from $O$. Each instance of the game will be evaluated as usual by the function~$V$. The players win the parallel repetition game if they win \emph{all} the instances.
More formally, for a game $G = (I,O,V,p)$ we define $G^n = (I',O',V',q)$, where
$I' = I^{\times n}, O' = O^{\times n}, q_{xy} = \Pi_{i \in [n]} p_{x_i,y_i}$ and $V'(a,b|x,y) = \Pi_{i \in [n]} V(a_i,b_i|x_i,y_i)$. While playing $G^n$, we say that Alice and Bob win game $i$ if $V(a_i,b_i|x_i,y_i) = 1$.

\subsubsection{Majority game}\label{Section:DefinitionMajorityGame}
For a game $G = (I,O,V,p)$ and a real number $\alpha \in [0,1]$ we define $G^n_\alpha = (I',O',V',p')$ as follows: $I' = I^{\times n}$, $O' = O^{\times n}$, $p'_{xy} = \Pi_{i \in [n]} p_{x_i,y_i}$ as in $G^n$. We define 
$
V'(a,b|x,y) = 1 \Leftrightarrow \#\{i: V(a_i,b_i|x_i,y_i) = 1\} \ge \alpha n.
$


\section{Advice states, superposed players and information cost}
\label{Section:FirstDefs}

The notion of information cost has been very useful for communication complexity. Here we derive a similar notion for entangled games.

Consider a game with advice state as defined in Section \ref{Section:advice}. The advice state can potentially greatly help the players. For example, Alice could know $y$ and Bob could know $x$. We ask ourselves the following question: \\
\emph{
For a game $G = (I=[k],O,V,p)$ such that $\omega^*(G) = 1 - \eps < 1$ and a state $\rho = \sum_{x,y \in [k]} p_{xy} \altketbra{x}_\spa{X} \otimes \altketbra{\phi_{xy}}_\spa{AB} \otimes \altketbra{y}_\spa{Y}$, what is the minimum dependency that the states $\{\ket{\phi_{xy}}\}_{xy}$ must have on $x,y$ to have $\omega^*(G|\rho) = 1$? 
} 

There are different ways of characterizing this dependency on $x,y$. A first possibility would be to consider the information that Alice has about $y$ and Bob has about $x$ while sharing $\rho$. However, there are cases where Alice and Bob can win a game with probability $1$ using an advice state while still not learning anything about each other's input. 
For example, take the CHSH game \cite{CHSH69} and consider the states $\ket{\phi_{00}} = \ket{\phi_{01}} = \ket{\phi_{10}} = \frac{1}{\sqrt{2}} (\ket{00} + \ket{11})_{\spa{AB}}$ and $\ket{\phi_{11}} = \frac{1}{\sqrt{2}} (\ket{01} + \ket{10})$. If the two players share the state $\rho = 
\sum_{x,y \in \zo} 1/4 \altketbra{x}_\spa{X} \otimes \altketbra{\phi_{xy}}_\spa{AB} \otimes \altketbra{y}_\spa{Y}$, Alice has no information about $y$ and Bob has no information about $x$. On the other hand, if both players measure their registers $\spa{A}$ and $\spa{B}$ in the computational basis and output the results, they will win the CHSH game with probability $1$ hence $\omega^*(CHSH|\rho) = 1$ while $\omega^*(CHSH) = \cos^2(\pi/8)$.

We must consider a slightly different scenario so that Alice or Bob can learn something about the other player's input. When considering the amount of information that Alice has about Bob's input $y$, we allow Alice to have a coherent superposition of her inputs. Similarly, we will be interested in the amount of information Bob has about $x$ when he has a coherent superposition of his inputs.

This scenario is motivated as follows: if Alice and Bob have a common procedure to create $\ket{\phi_{xy}}$ from their respective inputs $x$ and $y$, Alice can create a superposition of her inputs and they can perform the same procedure. This scenario has for example been in order to show optimal bounds for quantum bit commitment~\cite{CK11}.

This approach leads to the definition of the superposed information cost of a game. In the next section, we give formal definitions of this notion. 

\subsection{The superposed information cost}\label{Section:SuperposedDefs}

Consider a family of states $\{\ket{\phi_{xy}}\}_{xy}$ and a probability distribution $\{p_{xy}\}_{xy}$. Let $p_{x \cdot} = \sum_{y} p_{xy}$ and \mbox{$p_{\cdot y} = \sum_{x} p_{xy}$.}
 Let $\ket{L^B_x} = \frac{1}{\sqrt{p_{x \cdot}}} \sum_{y} \sqrt{p_{xy}} \ket{\phi_{xy}} \ket{y}$ and $\ket{L^A_y} = \frac{1}{\sqrt{p_{\cdot y}}} \sum_{x} \sqrt{p_{xy}} \ket{x} \ket{\phi_{xy}}$. Consider the two superposed states:
\begin{align*}
\sigma^A = \sum_{y \in [k]} p_{\cdot y} \altketbra{L^A_y}_{\spa{XAB}} \otimes \altketbra{y}_{\spa{Y}} \\
\sigma^B = \sum_{x \in [k]} p_{x \cdot} \altketbra{x}_{\spa{X}} \otimes \altketbra{L^B_x}_{\spa{ABY}}.
\end{align*}

Here $\sigma^A$ (resp. $\sigma^B$) corresponds to $\rho$ where Alice's input (resp. Bob's input) is put in a coherent superposition. We first define the superposed information cost of a family of states with a probability distribution.
\begin{definition}
The superposed information cost $SIC(\{\ket{\phi_{xy}},p_{xy}\}_{xy})$ is defined as 
$
SIC(\{\ket{\phi_{xy}},p_{xy}\}_{xy}) = I(Y:XA)_{\sigma^A} + I(X:BY)_{\sigma^B}. $
\end{definition}
\textbf{Remark: }{This definition has good properties when the input distribution is a product distribution or close to a product distribution. One may want to consider a more general definition when considering any distribution.}

We also define the superposed information cost of a shared state $\rho$ of the form $\rho = \sum_{xy \in [k]} p_{xy} \altketbra{x} \otimes \altketbra{\phi_{xy}} \otimes \altketbra{y}$.
\begin{definition}
$SIC(\rho) = \inf\{SIC(\{\ket{\phi_{xy}},p_{xy}\}_{xy})\}$
where the infimum is taken over all families $\{\ket{\phi_{xy}},p_{xy}\}_{xy}$ s.t.\ $\rho = \sum_{xy \in [k]} p_{xy} \altketbra{x} \otimes \altketbra{\phi_{xy}} \otimes \altketbra{y}$.
\end{definition}
\textbf{Remark:} Notice that a state $\rho$ doesn't uniquely define states $\{\ket{\phi_{xy}},p_{xy}\}_{xy}$ because it doesn't capture the phases in the states $\ket{\phi_{xy}}$.

We now define the superposed information cost of an entangled game.
\begin{definition}
For any entangled game $G = (I,O,V,p)$,  we define $SIC(G) = \inf \{SIC(\{\ket{\phi_{xy}}\}_{xy},\{p_{xy}\}_{xy})\}$ where the infimum is taken over all $(\{\ket{\phi_{xy}}\}_{xy},\{p_{xy}\}_{xy})$ such that the associated state $\rho = \sum_{xy} p_{xy} \altketbra{x} \otimes \altketbra{\phi_{xy}} \otimes \altketbra{y}$ satisfies $ \omega^*(G|\rho) = 1$.
\end{definition}

The superposed information cost behaves nicely under parallel repetition. In Appendix \ref{Appendix:Additivity}, we show
\begin{proposition}
For any game $G$, we have $SIC(G^n) = n \cdot SIC(G)$.
\end{proposition}


\addtocounter{theorem}{-1}
\section{Organisation of the proof of Theorem~\ref{Theorem:Final}}\label{Section:Orga}
In Section~\ref{Section:Single}, we show how to use the \emph{Superposed Information Cost} of a game $G$ to bound its entangled value $\omega^*(G)$. We first show:
\begin{theorem}\label{Theorem:OneShotPerfect}
For any game $G$ on the uniform distribution, 
$SIC(G) \ge \frac{1 - \omega^*(G)}{32\ln(2)}$.
\end{theorem} 
We also extend this theorem as follows:
\begin{theorem}\label{Theorem:OneShot}\label{Theorem:SingleGame}
There exists a small constant $c_0$ such that for any game $G = (I = [k],O,V,\mbox{Unif.})$ satisfying $\omega^*(G) = 1 - \eps$, for any game $G' = (I = [k],O,V,p)$ satisfying $\frac{1}{2} \sum_{x,y} |p_{xy} - \frac{1}{k^2}| \le c_0{\eps}$ and any state $\rho = \sum_{xy} p_{xy} \altketbra{x} \otimes \altketbra{\phi_{xy}} \otimes \altketbra{y}$ such that $\omega^*(G'|\rho) \ge 1 - \frac{\eps}{4}$, we have that $SIC(\rho) = \Omega(\eps)$. 
\end{theorem}

If $\omega^*(G) = 1 - \eps$, Theorem~\ref{Theorem:OneShotPerfect} claims that $SIC(G) \ge \frac{\eps}{32 \ln(2)}$ which gives by additivity of the superposed information cost that $SIC(G^n) \ge \frac{n \eps}{32 \ln(2)}$. Ideally, we would like to upper bound $SIC(G^n)$ with a function of $\omega^*(G^n)$. Unfortunately, we are not able to do this directly. In Section~\ref{Sec:UpperBound}, we show the following weaker statement:

\begin{theorem}\label{Theorem:UpperBound}
Consider a game $G = (I,O,V,\mbox{Unif.})$ such that $\omega^*(G) = 1 - \eps$ and $\omega^*(G^n) = 2^{-t}$. Let $G^n_{1 - \eps/32} = (I^n,O^n,V',\mbox{Unif.})$ as defined in Section~\ref{Section:DefinitionMajorityGame}. There exists a game $G' = (I^n,O^n,V',p)$ and a state $\xi = \sum_{xy} p_{xy} \altketbra{x} \otimes \altketbra{\phi_{xy}} \otimes \altketbra{y}$
satisfying the following properties:
\begin{enumerate}
\item $H(XY)_{\xi} \ge 2n\log(k) - t - 1$
\item $\omega^*(G'|\xi) \ge 1 - \eps/32$ 
\item $SIC(\xi) \le \frac{32 \log(|I||O|)}{\eps} ((t+1) + |\log(\eps)| + 5) +2t + 2$.
\end{enumerate}
\end{theorem}

The first condition states that $p$ is in some sense close to the uniform distribution hence $G'$ is close to $G^n_{1 - \eps/32}$. This theorem is weaker than an upper bound on $SIC(G')$ which itself is weaker than an upper bound on $SIC(G^n)$, but this kind of upper bound will be enough.

In Appendix \ref{Appendix:LowerBound}, we prove the following matching lower bound.
\begin{theorem}\label{Theorem:LowerBound}
Consider a game $G = (I = [k],O,V,\mbox{Unif.})$ such that $\omega^*(G) = 1 - \eps$ and $\omega^*(G^n) = 2^{-t}$ with $t = o(n\eps)$. Let also $G^n_{1-\eps/32} = (I^n = [k^n],O^n,V',\mbox{Unif.})$ as defined in Section~\ref{Section:DefinitionMajorityGame}. For any game $G'=(I' = [k^n],O',V',p)$ and any state $
\rho = \sum_{x,y \in [k^n]} p_{xy} \altketbra{x}_{\spa{X}} \otimes \altketbra{\phi_{xy}}_{\spa{AB}} \otimes \altketbra{y}_\spa{Y},
$ satisfying
\begin{enumerate}
\item $H(XY)_\rho \ge 2n\log(k) - t - 1$
\item $\omega^*(G'|\rho) \ge 1 - \eps/32$
\end{enumerate}
we have 
$SIC(\rho) \ge \Omega(n  \eps)$.
\end{theorem}

In Section~\ref{Section:finalThm}, we show how to use the two above theorems to conclude:

\addtocounter{theorem}{-5}
\begin{theorem}\label{Theorem:Final}
For any game $G = (I,O,V,\mbox{Unif.})$ with $\omega^*(G) \le 1 - \eps$, we have \\ $\omega^*(G^n) = (1 - \eps^2)^{\Omega\left(\frac{n}{\log(|I||O|)} - |\log(\eps)|\right)}$.
\end{theorem}


\section{Overview of Theorem \ref{Theorem:OneShotPerfect}}\label{Section:Single}
\begin{theorem}\label{Theorem:SingleGamePerfect}
For any game $G$ on the uniform distribution,
$SIC(G) \ge \frac{1 - \omega^*(G)}{32\ln(2)}$.
\end{theorem}

We sketch the proof as follows. 
We fix a game $G = (I=[k],O,V,\mbox{Unif.})$ and a state $\rho = \sum_{x,y} \frac{1}{k^2} \altketbra{x}_\spa{X} \otimes \altketbra{\phi_{xy}}_\spa{AB} \otimes \altketbra{y}_\spa{Y}$ such that $\omega^*(G|\rho) = 1$. 
As in Section~\ref{Section:SuperposedDefs}, we define $\ket{L^A_y},\ket{L^B_x},\sigma^A,\sigma^B$. Let $\rho^A_y = \Tr_{\spa{B}} \ketbra{L^A_y}{L^A_y}$ and $\rho^B_x = \Tr_{\spa{A}} \ketbra{L^B_x}{L^B_x}$. Intuitively, $\rho^A_y$ (resp.\ $\rho^B_x$) corresponds to the input-superposed state that Alice (resp.\ Bob) has, conditioned on Bob getting $y$ (resp.\ Alice getting $x$).
Let $F$ denote the fidelity of quantum states.
We prove the following  three inequalities.
\begin{enumerate}
\item First we show that $SIC(\rho) \ge \frac{1}{4\ln(2)} (1 - \frac{1}{k^2} \sum_{y,y'} F^2(\rho^A_y,\rho^A_{y'}) + 1 - \frac{1}{k^2} \sum_{x,x'} F^2(\rho^B_x,\rho^B_{x'}))$ 
\item Then we show that  $$ 1 - \frac{1}{k^2} \sum_{y,y'} F^2(\rho^A_y,\rho^A_{y'}) + 1 - \frac{1}{k^2} \sum_{x,x'} F^2(\rho^B_x,\rho^B_{x'}) \ge \frac{1}{8}(1 - \max_{\ket{\Omega}}\sum_{x,y \in [k]} \frac{1}{k^2} |\triple{\Omega}{U_x \otimes V_y}{\phi_{xy}}|^2)$$ for some (sets of) unitaries $\{U_x\}_x,\{V_y\}_y$. 
\item Finally, we show that $(1 - \max_{\ket{\Omega}}\sum_{x,y \in [k]} \frac{1}{k^2} |\triple{\Omega}{U_x \otimes V_y}{\phi_{xy}}|^2) \ge 1 - \omega^*(G)$. 
\end{enumerate}
Putting the three inequalities together, we get 
\begin{align*}
SIC(\rho) & \ge \frac{1}{4\ln(2)} (1 - \frac{1}{k^2} \sum_{y,y} F^2(\rho^A_y,\rho^A_{y'}) + 1 - \frac{1}{k^2} \sum_{x,x'} F^2(\rho^B_x,\rho^B_{x'})) \\
& \ge \frac{1}{32\ln(2)}(1 - \max_{\ket{\Omega}}\sum_{x,y \in [k]} \frac{1}{k^2} |\triple{\Omega}{U_x \otimes V_y}{\phi_{xy}}|^2) \quad \mbox{for some } \{U_x\}_x \{V_y\}_y \\
& \ge \frac{1 - \omega^*(G)}{32\ln(2)}.
\end{align*}
Since this holds for any $\rho$ satisfying $\omega^*(G|\rho) = 1$, we have $SIC(G) \ge \frac{1 - \omega^*(G)}{32\ln(2)}$.

%


\section{Overview of Theorem \ref{Theorem:UpperBound}} \label{Sec:UpperBound}
In this section we sketch the proof of Theorem~\ref{Theorem:UpperBound}. The construction of the state $\xi$ will directly be inspired by a communication task that we now present.
\subsubsection*{The communication task}\label{Section:CommunicationProtocol}
Fix a game $G = (I,O,V,\mbox{Unif.})$ satisfying $\omega^*(G) = 1 - \eps$. Let $G^n = (I^n,O^n,V_n,\mbox{Unif.})$ such that $\omega^*(G^n) = 2^{-t}$ for some $t$. We now consider the following task $H(p,m)$. \\ \\
\cadre{
\begin{center} Task $\textrm{H}(p,m)$ \end{center}
\begin{itemize}
\item Alice and Bob are allowed to share any quantum state $\ket{\phi}$.
\item Alice and Bob get inputs $x = x_1,\dots,x_n$ and $y = y_1,\dots,y_n$, with $x,y \in I^n$, following the uniform distribution.
\item Alice is allowed to send $m$ bits to Bob
\item Then Alice outputs some value $a \in O^n$ and Bob outputs some value $b \in O^n$ or 'Abort'.
\end{itemize}

For each index $i$, we say that Alice and Bob win game $i$ if Bob does not abort and $V(a_i,b_i|x_i,y_i) = 1$. We require the following
\begin{enumerate}
\item $Pr[\textrm{Bob does not abort}] \ge p$
\item $Pr[\textrm{Alice and Bob win } \ge (1-\eps/32)n \textrm{ games }| \textrm{ Bob does not abort} ] \ge (1-\eps/32)$.
\end{enumerate}
}
\\ \ \\ 

Showing how to perform this task with a small amount of communication is a first step towards the construction of $\xi$.
We consider the following protocol $P$ that efficiently performs this task. \\ \\
\cadre{\begin{center}Protocol $P$ for the task H(p,m)\end{center} 
\begin{enumerate}
\item Let $v\leq n$ be an integer, to be determined at the end of this section. Alice and Bob have  shared randomness that correspond to $v$ random (not necessarily different) indices $i_1,\dots,i_v \in [n]$ as well as a state $\ket{\phi}$ that allows them to win $G^n$ with probability at least $\frac{\omega^*(G^n)}{2} = 2^{-(t+1)}$.
\item Alice and Bob receive uniform inputs $x,y$. They perform a strategy that wins all $n$ games with probability $2^{-(t+1)}$ and have some outputs $a = a_1,\dots,a_n$ and $b = b_1,\dots,b_n$.
\item For each index $i \in \{i_1,\dots,i_v\}$, Alice sends ${x_i}$ and ${a_i}$ to Bob.
\item For each of these indices $i$, Bob looks at $x_i,y_i,a_i,b_i$ and checks whether they win on all of these $v$ games, \ie, he checks that for all these indices, $V(a_i,b_i|x_i,y_i) = 1$.
\item If they do win on all of these games, Bob outputs $b$. Otherwise, Bob outputs 'Abort'.
\end{enumerate}
}

\begin{proposition}~\label{Proposition:ProtocolEfficiency}
The above protocol performs the task $H(p,m)$ with $p \ge 2^{-(t+1)}$ and $m = \frac{32 \log(|I||O|)}{\eps} ((t+1) + |\log(\eps)| + 5)$.
\end{proposition}

\begin{proof}
We have:
\begin{align*}
\Pr[\textrm{Bob does not abort}] & = \Pr[\textrm{Alice and Bob win } G_i \ \forall i \in \{i_1,\dots,i_v\}] \\ & \ge \Pr[\textrm{Alice and Bob win } G_i \ \forall i \in [n] ] = 2^{-(t+1)},
\end{align*}
hence $p \ge 2^{-(t+1)}$. 
For a uniformly random index $i$, we have:
\begin{align*}
\Pr[\textrm{Alice and Bob win } G_i | \textrm{ Alice and Bob win }  \le (1 - \eps/32)n \textrm { games }] \le 1 - \eps/32.
\end{align*}
Since the indices in $\{i_1,\dots,i_v\}$ are independent random indices in $[n]$, we~have 
\begin{align*}
& \Pr[\textrm{Bob does not abort } | \textrm{ Alice and Bob win } \le (1 - \eps/32)n \textrm { games}] \\
& =  \Pr[\textrm{Alice and Bob win } G_i \ \forall i \in \{i_1,\dots,i_v\} | \textrm{ Alice and Bob win } \le (1 - \eps/32)n \textrm { games}] \\ 
& \le (1 - \eps/32)^v.
\end{align*}
Next, we have:
\begin{align*}
& \Pr[\textrm{A and B win } \le (1 - \eps/32)n \textrm { games} \mid \textrm{B does not abort}] \cdot \Pr[\textrm{B does not abort}] \\
& =  \Pr[\textrm{B does not abort} \mid \textrm{A and B win }  \le (1 - \eps/32)n \textrm { games}] \cdot \Pr[\textrm{A and B win }  \le (1 - \eps/32)n \textrm { games}] \\
& \le \Pr[\textrm{B does not abort} \mid \textrm{A and B win }  \le (1 - \eps/32)n \textrm { games}] \\
& \le (1 - \eps/32)^v.
\end{align*}
This gives us:
\begin{align*}
\Pr[\textrm{A and B win } \le (1 - \eps/32)n \textrm { games } | \textrm{ B does not abort}] & \le \frac{(1 - \eps/32)^v}{\Pr[\textrm{B does not abort}]} \\
& \le \frac{(1 - \eps/32)^v}{2^{-(t+1)}}.
\end{align*}
We can take $v = \frac{32}{\eps}((t+1) + |\log(\eps)| + 5)$, such that we have
$$
\Pr[\textrm{A and B win } \le (1 - \eps/32)n \textrm { games } | \textrm{ B does not abort}] \le \eps/32.
$$
Notice that $m = v \cdot \log(|I||O|)$.  Therefore, if Alice sends $m = \frac{32 \log(|I||O|)}{\eps} ((t+1) + |\log(\eps)| + 5)$ bits to Bob, 
\begin{align*}
\Pr[\textrm{A and B win } \ge (1 - \eps/32)n \textrm { games } | \textrm{ B does not abort}] & \ge 1 - \eps/32.
\end{align*}
\qed 
\end{proof}

\subsection*{Using the communication task to prove Theorem \ref{Theorem:UpperBound}}
The idea is the following: Alice and Bob perform protocol P for the task $H(p,m)$ performing everything in superposition, including the messages and their shared randomness. The advice state we consider is the state $\rho_{NA}$ Alice and Bob share conditionned on Bob not aborting. This state $\rho_{NA}$ can be written as
\[
\rho_{NA} = \sum_{xy} q_{xy} \altketbra{x}_\spa{X} \otimes \altketbra{\phi_{xy}} \otimes \altketbra{y}_\spa{Y}
\]

 To prove the theorem, we must show the following properties for $\rho_{NA}$.
\begin{enumerate}
\item $H(XY)_{\rho_{NA}} \ge 2n\log(k) - t - 1.$
\item $\omega^*(G'|\rho_{NA}) \ge 1 - \eps/32$  \textnormal{ where } $G^n_{1 - \eps/32} = (I',O',V',\mbox{\textnormal{Unif.}})$ \textnormal{ and } $ G' = (I',O',V',q).$
\item $SIC(\rho_{NA}) \le \frac{32 \log(|I||O|)}{\eps} ((t+1) + |\log(\eps)| + 5) + 2t + 2.$
\end{enumerate}
The ideas behind the proofs of these three properties are as follows:
\begin{enumerate}
\item In task H(p,m), $Pr[\textrm{Bob does not abort}] \ge p= 2^{-t}$, when conditionning on Bob winning, we remove at most t bits of entropy from the (uniform) inputs in $X,Y$ , the 1 in the inequality is there for technical reasons. 

\item In the task H(p,m), $Pr[\textrm{Alice and Bob win } \ge (1-\eps/32)n \textrm{ games }| \textrm{ Bob does not abort} ] \ge (1-\eps/32)$. This directly implies the second property

\item In protocol P, before Alice sends her message, Bob has no information about $x$. Alice sends a message of size $m$, which gives $m$ bits of information about Alice's input. Conditionning on Bob winning gives him an extra $2t$ bits of information. Since $m = \frac{32 \log(|I||O|)}{\eps} ((t+1) + |\log(\eps)| + 5)$ from the previous Proposition, we can conclude.
\end{enumerate}




\section{Final Theorem} \label{Section:finalThm}
\addtocounter{theorem}{-2}
\begin{theorem}
For any game $G = (I,O,V,\mbox{Unif.})$ with $\omega^*(G) \le 1 - \eps$, we have: $$\omega^*(G^n) = (1 - \eps^2)^{\Omega\left(\frac{n}{\log(|I||O|)} - |\log(\eps)|\right)}.$$
\end{theorem} 
\begin{proof}
Let $G^n_{1-\eps/32} = (I^n = [k^n],O^n,V_n,\mbox{Unif.})$ as defined in Section~\ref{Section:DefinitionMajorityGame}.
Using Theorem 4, we know there exists a state $\xi = \sum_{xy} p_{xy} \altketbra{x} \otimes \altketbra{\phi_{xy}} \otimes \altketbra{y}$ and a game $G' = (I^n,O^n,V_n,p)$ satisfying 
\begin{enumerate}
\item $H(XY)_{\xi} \ge 2n\log(k) - t - 1$
\item $\omega^*(G'|\xi) \ge 1 - \eps/32$ 
\item $SIC(\xi) \le \frac{32 \log(|I||O|)}{\eps} ((t+1) + |\log(\eps)| + 5),$
\end{enumerate}
where $2^{-t} = \omega^*(G^n)$. We now distinguish two cases
\begin{itemize}
\item If $t = \Omega(\eps n)$ then $\omega^*(G^n) = (1-\eps)^{\Omega(n)}$ and the theorem holds directly.
\item If $t = o(\eps n)$, we need the following argument. The state $\xi$ satisfies all the properties of Theorem~\ref{Theorem:LowerBound} which implies that
$
SIC(\xi) = \Omega(n\eps)
$. We combine the two inequalities and obtain
\begin{align*}
\Omega(n\eps) \le SIC(\xi) \le \frac{32 \log(|I||O|)}{\eps} ((t+1) + |\log(\eps)| + 5).
\end{align*}
It follows that
$t = \Omega\left(\frac{n \eps^2}{\log(|I||O|)} - |\log(\eps)|\right)$, which allows us to conclude
$$
\omega^*(G^n) = 2^{-t} \le (1 - \eps^2)^{O\left(\frac{n}{\log(|I||O|)} - |\log(\eps)|\right)}.
$$ \qed
\end{itemize}
\end{proof}

\addtocounter{corollary}{-1}
Finally, we extend the result to games with complete support (\ie, games on distributions $p$ such that $\not\exists (x,y)$ for which $p_{xy} = 0$). This bound is weaker than the main result, because it depends also on $p$.
\begin{corollary}
Let $G=(I,O,V,p)$ be a game with complete support and $\omega^*(G) = (1 - \eps)$. Then,
$$\omega^*(G^n) \le (1 - \eps^2)^{\Omega(\frac{n}{Q\log(|I||O|)} - \frac{|\log(\eps)|}{Q})},$$
where $Q= \frac{k^2 \max_{xy} p_{xy}^2 }{\min_{xy} p_{xy} }$.
\end{corollary}
The proof of the above Corollary is in Appendix \ref{Appendix:GeneralGames}.

\subsubsection*{Acknowledgments}
The authors thank Serge Fehr, Anthony Leverrier, Christian Schaffner and Ronald de Wolf for helpful suggestions. Part of the work was done while A.C. was at CWI, Amsterdam. A.C. was partially supported by the European Commission under the project QCS (Grant No.~255961). G.S. was supported by de Wolf's Vidi grant 639.072.803 from the Netherlands Organization for Scientific Research (NWO).

\bibliography{paper}
\bibliographystyle{plain}

\COMMENT{

}

%
%
%
%
%
%

\newpage
\begin{appendix}
\section{Preliminaries}
\label{Section:Prel}
\subsection{Useful facts about the fidelity and trace distance of two quantum states.}\label{HowClose}

We start by stating a few properties of the trace distance $\Delta$ and fidelity $F$ between two quantum states. These two notions characterize how close two quantum states are.

\subsubsection{Trace distance between two quantum states}

\begin{definition}
For any two quantum states $\rho,\sigma$, the trace distance $\Delta$ between them is given by $\Delta(\rho,\sigma) = \Delta(\sigma,\rho) = \frac{1}{2}\tnorm{\rho - \sigma}$.
\end{definition}

Here the used trace norm may be expressed as 
$
  \tnorm{X} = \sqrt{X^\dagger X} = \max_{U} | \mathrm{tr} (X U) |,
$
where the maximization is taken over all unitaries of the appropriate
size.

\begin{proposition}~\label{PropPOVM}
For any two states $\rho,\sigma$, and a POVM $E = \{E_1,\dots,E_m\}$ with $p_i = tr(\rho E_i)$ and $q_i = tr(\sigma E_i)$, we have $\Delta(\rho,\sigma) \ge \frac{1}{2} \sum_{i} |p_i - q_i|$. There exists a POVM (even a projective measurement) for which this inequality is an equality.
\end{proposition}

\begin{proposition}~\cite{Hel67}\label{Hel67}
Suppose Alice has a uniformly random bit $c \in \zo,$ unknown to Bob. She sends a quantum state $\rho_c$ to Bob. We have
\[
\Pr[\mbox{Bob guesses } c] \le \frac{1}{2} + \frac{\Delta(\rho_0,\rho_1)}{2}.
\]
There is a strategy for Bob that achieves the value $\frac{1}{2} + \frac{\Delta(\rho_0,\rho_1)}{2}$.
\end{proposition}

\subsubsection{Fidelity of quantum states}

\begin{definition}
For any two states $\rho,\sigma$, their fidelity $F$ is given by
$F(\rho,\sigma) = F(\sigma,\rho)= tr(\sqrt{\rho^{\frac{1}{2}}\sigma\rho^{\frac{1}{2}}}) $
\end{definition}

\begin{proposition}\label{POVMfidelity}
For any two states $\rho,\sigma$, and a POVM $E = \{E_1,\dots,E_m\}$ with $p_i = tr(\rho E_i)$ and $q_i = tr(\sigma E_i)$, we have $F(\rho,\sigma) \le \sum_{i} \sqrt{p_i q_i}$. There exists a POVM for which this inequality is an equality.
\end{proposition}

\begin{definition}
We say that a pure state $\ket{\psi}$ in $\spa{A} \otimes \spa{B}$ is a purification of some state $\rho$ in $\spa{B}$ if $\ \Tr_{\spa{A}}(\ketbra{\psi}{\psi}) = \rho$.
\end{definition}

\begin{proposition}[Uhlmann's theorem]
For any two quantum states $\rho,\sigma$, there exists a purification $\ket{\phi}$ of $\rho$ and a purification $\ket{\psi}$ of $\sigma$ such that $|\braket{\phi}{\psi}| = F(\rho,\sigma)$.
\end{proposition}

\begin{proposition}\label{CPTPfidelity}
For any two quantum states $\rho,\sigma$ and a completely positive trace preserving operation $Q$, we have $F(\rho,\sigma) \le F(Q(\rho),Q(\sigma))$.
\end{proposition}

\begin{proposition}[\cite{SR01,NS03}]~\label{Prop:FidelityInequality}
For any two quantum states $\rho, \sigma$ 
\[
    \max_{\xi} \left( F^2(\rho, \xi) + F^2(\xi, \sigma) \right)
    = 1 + \F(\rho,\sigma).
\]
\end{proposition}

\begin{proposition}[\cite{FV99}]\label{FuchsVdG}
For any quantum states $\rho,\sigma$, we have
\[
1 - F(\rho,\sigma) \le \Delta(\rho,\sigma) \le \sqrt{1 - F^2(\rho,\sigma)}.
\]
\end{proposition}

As direct corollaries of Proposition~\ref{Prop:FidelityInequality}, we have 
\begin{proposition}\label{Proposition:Trig}
Let $\ket{A},\ket{B},\ket{C}$ be three quantum states. We have 
\begin{align*}
|\braket{A}{C}| \ge |\braket{A}{B}|^2 + |\braket{B}{C}|^2 - 1.
\end{align*}
\end{proposition}
and 
\begin{proposition}\label{Prop:FidelityInequality2}
For any 3 quantum states $\rho_1,\rho_2,\rho_3$, we have 
\begin{align*}
(1 - F(\rho_1,\rho_2)) + (1 - F(\rho_2,\rho_3)) \ge \frac{1}{2} (1 - F(\rho_1,\rho_3)),
\end{align*}
or equivalently
$
F(\rho_1,\rho_3) \ge 1 - 2(1 - F(\rho_1,\rho_2) + 1 - F(\rho_2,\rho_3)).
$
\end{proposition}
\begin{proof}
Using Proposition~\ref{Prop:FidelityInequality}, we have
\begin{align*}
1 + F(\rho_1,\rho_3) & = \max_{\xi} \left( F^2(\rho_1, \xi) + F^2(\xi, \rho_3) \right) \\
& \ge  F^2(\rho_1, \rho_2) + F^2(\rho_2, \rho_3),
\end{align*}
which gives
\begin{align*}
1 - F(\rho_1,\rho_3) \le 1 - F^2(\rho_1,\rho_2) + 1 - F^2(\rho_2,\rho_3) \le 2(1 - F(\rho_1,\rho_2)) + 2(1 - F(\rho_2,\rho_3)).
\end{align*}
Hence
$
1 - F(\rho_1,\rho_2) + 1 - F(\rho_2,\rho_3) \ge \frac{1}{2} (1 - F(\rho_1,\rho_3)).
$
\qed \end{proof}

\begin{proposition}\label{Prop:FidelityLast}
For two quantum states $\rho = \sum_{x} p_x \altketbra{x} \otimes \rho_x$ and $\rho' = \sum_x p'_x \altketbra{x} \otimes \rho'_x$, we have $F(\rho,\rho') = \sum_x \sqrt{p_x p_{x'}} F(\rho_x,\rho_{x'})$.
\end{proposition}
\begin{proof}
We use the following definition  of the fidelity: $F(\rho,\rho') = ||\sqrt{\rho}\sqrt{\rho'}||_1$. From there, we immediately have that 
\[ F(\rho,\rho') = \sum_{x} \sqrt{p_x p_{x'}} ||\sqrt{\rho_x}\sqrt{\rho'_x}||_1 = \sum_{x} \sqrt{p_x p_{x'}} F(\rho_x,\rho_{x'}).\]
\qed \end{proof}

\subsection{Information Theory}\label{Section:InformationTheory}
For a quantum state $\rho$, the entropy of $\rho$ is $
H(\rho) = - tr(\rho \log(\rho))$.
For a quantum state $\rho \in \spa{X} \otimes \spa{Y}$, $H(X)_\rho$ is the entropy of the quantum register in the space $\spa{X}$ when the total underlying state is $\rho$. In other words, $H(X)_\rho = H(\Tr_{\spa{Y}} (\rho))$.

$H(X|Y)_{\rho} = H(XY)_{\rho} - H(Y)_{\rho}$ is the conditional entropy of $X$ given $Y$ on $\rho$ and $I(X:Y)_\rho = H(X)_\rho + H(Y)_\rho - H(XY)_\rho $ is the mutual information between $X$ and $Y$ on $\rho$.

We define $H_{min}(\rho) = -\log(\lambda_{max})$ where $\lambda_{\max}$ is the maximum eigenvalue of $\rho$. For $\rho$ in $\spa{X} \otimes \spa{Y}$, we define 
\begin{align*}
H_{min}(X|Y)_\rho = \max_{\sigma \in Y} \sup\{\lambda: \rho \le  2^{- \lambda} I_{\spa{X}} \otimes \sigma\}
\end{align*}
We have $H_{min}(X|Y)_\rho \le H(X|Y)_\rho$~\cite{MDS+13}.
In the case where Alice and Bob share $\rho = \sum_{x} p_x \altketbra{x}_\spa{X} \otimes \rho(x)_\spa{Y}$, where Alice has register $\spa{X}$ and Bob has register $\spa{Y}$, we have 
$H_{min}(X|Y)_\rho = - \log(\Pr[\textrm{Bob can guess x}])$. 

\begin{claim}[Subadditivity of the conditional entropy]
\begin{align*}
H(AB|C) \le H(A|C) + H(B|C)
\end{align*}
\end{claim}

\begin{claim}[\cite{KNTZ07}]\label{Claim:DC}
\[
I(A:B)_{\rho} \ge \frac{2}{\ln(2)} (1 - F(\rho,\rho_A \otimes \rho_B))
\]
where $\rho_A = Tr_{\spa{B}} (\rho)$ and $\rho_B = Tr_{\spa{A}} (\rho)$
\end{claim}

\begin{claim}[from \cite{Vad99}]\label{Vadhan}
For any distribution $p$ on a universe $U$, if $H(p) \ge \log(|U|) - \eps$ then $\Delta(p,\mbox{Unif.}) \le \eps$, where $\mbox{Unif.}$ is the uniform distribution.
\end{claim}

\section{Additivity of the superposed information cost}\label{Appendix:Additivity}
Our goal here is to prove the additivity of the superposed information cost, $\ie$ that $SIC(G^n) = nSIC(G)$. Before the proof, we introduce some notation and prove a lemma. 

Let $G = (I,O,V,p)$ and let $G^n = (I^n,O^n,V_n,q)$. 
For a string $x = x_1,\dots,x_n \in I^n$, let $x_{-i}$ be the string in $I^{n-1}$ where we remove $x_i$ from $x$. 
Let $\rho = \sum_{x,y \in I^n} q_{xy} \altketbra{x} \otimes \altketbra{\phi_{xy}} \otimes \altketbra{y}$ satisfying $\omega^*(G^n|\rho) = 1$. As in Section~\ref{Section:SuperposedDefs}, we define $\ket{L^A_y},\ket{L^B_x},\sigma^A,\sigma^B$ for $\rho$. We first prove the following Lemma:
\begin{lemma} \label{Lemma:SIC_i}
For all $i\in[n]$ we have that
\begin{align*} 
I(Y_i:XA)_{\sigma^A} + I(X_i:YB)_{\sigma^B} \ge SIC(G).
\end{align*}
\end{lemma}

\begin{proof}
By definition of $G^n$, we have $q_{xy} = \Pi_j p_{x_j,y_j}.$ We define $q^{-i}_{xy} = \Pi_{j \neq i} p_{x_j,y_j}.$
For each $i$, we can rewrite $\rho$ as:
\begin{align*}
\rho = \sum_{x,y \in I^n} q_{xy} \altketbra{x_i}_{\spa{X}_i} \otimes \altketbra{x_{-i}}_{\spa{X}_{-i}} \otimes \altketbra{\phi_{xy}}_{AB} \otimes \altketbra{y_{-i}}_{\spa{Y}_{-i}} \otimes \altketbra{y_i}_{\spa{Y}_i}.
\end{align*}
We define 
\[
\ket{Z^i_{x_i,y_i}} = \sum_{x',y' \in I^n : x'_i = x_i, y'_i = y_i} \sqrt{q^{-i}_{x'y'}} \ket{x'_{-i}} \otimes \ket{\phi_{x'y'}} \otimes \ket{y'_{-i}}.
\]

Let $\rho_i = \sum_{x_i,y_i \in I} p_{x_i,y_i} \altketbra{x_i} \otimes \altketbra{Z^i_{x_i,y_i}} \otimes \altketbra{y_i}$.  $\rho_i$ corresponds to $\rho$ where the registers in $\spa{X}_{-i},\spa{Y}_{-i}$ are put in superposition. Hence, Alice and Bob can go from $\rho_i$ to $\rho$ by measuring the registers $\spa{X}_{-i}$ and $\spa{Y}_{-i}$ in the computational basis. 
Using $\rho$, Alice and Bob can win the $i^{th}$ instance of $G$ with probability $1$. This means that they can also win this $i^{th}$ instance of $G$ when sharing $\rho_i$ and $\omega^*(G|\rho_i) = 1$.

We define
\begin{align*}
\ket{L^B_{x_i}(i)} & = \frac{1}{\sqrt{p_{x_i \cdot}}} \sum_{y_i \in I} \sqrt{p_{x_i,y_i}} \ket{Z^i_{x_i,y_i}} \ket{y_i} \\ \ket{L^A_{y_i}(i)} & = \frac{1}{\sqrt{p_{\cdot y_i}}}  \sum_{x_i \in I}  \sqrt{p_{x_i,y_i}} \ket{x_i} \ket{Z^i_{x_i,y_i}}.
\end{align*} We now also define the two new superposed states of $\rho_i$
\begin{align*}
\sigma^B_i = \sum_{x_i \in I} p_{x_i \cdot} \altketbra{x_i}_{\spa{X}_i} \otimes \altketbra{L^B_{x_i}(i)}_{\spa{X}_{-i}\spa{ABY}} \\
\sigma^A_i = \sum_{y_i \in I} p_{\cdot y_i} \altketbra{L^A_{y_i}(i)}_{\spa{XAB}\spa{Y}_{-i}} \otimes \altketbra{y_i}_{\spa{Y}_i}.
\end{align*}
$\omega^*(G|\rho_i) = 1$ implies $SIC(\rho_i) \ge SIC(G)$ hence
\begin{align*}
I(Y_i:XA)_{\sigma^A_i}  + I(X_i:BY)_{\sigma^B_i} \ge SIC(G).
\end{align*}
$\sigma^A_i$ corresponds to $\sigma^A$ where the input registers $\spa{Y}_{-i}$ are put in a coherent superposition. From there, we have
$Tr_{\spa{Y}_{-i}}(\sigma^A_i) = Tr_{\spa{Y}_{-i}} (\sigma^A)$ and $I(Y_i:XA)_{\sigma^A_i} = I(Y_i:XA)_{\sigma^A}$. Similarly, we have  $I(X_i:BY)_{\sigma^B_i} = I(X_i:BY)_{\sigma^B}$, which gives
\begin{align*}
I(Y_i:XA)_{\sigma^A}  + I(X_i:YB)_{\sigma^B} \ge SIC(G).
\end{align*}
\qed \end{proof}

We can now prove our proposition:

\addtocounter{proposition}{-12}
\begin{proposition}\label{Prop2}
$SIC(G^n) = n SIC(G)$.
\end{proposition}
\addtocounter{proposition}{+11}

\begin{proof}
We have:
\begin{align*}
SIC(\rho) & = I(Y:XA)_{\sigma^A} + I(X:BY)_{\sigma^B} \\
& = H(Y)_{\sigma^A} - H(Y|XA)_{\sigma^A}  + H(X)_{\sigma^B} - H(X|BY)_{\sigma^B} \\
& = \sum_{i \in [n]} H(Y_i)_{\sigma^A} - H(Y|XA)_{\sigma^A} + \sum_{i \in [n]} H(X_i)_{\sigma^B} - H(X|BY)_{\sigma^B}  \\  
& \ge \sum_{i \in [n]} H(Y_i)_{\sigma^A} - \sum_{i \in [n]} H(Y_i|XA)_{\sigma^A} + \sum_{i \in [n]} H(X_i)_{\sigma^B} - \sum_{i \in [n]} H(X_i|BY)_{\sigma^B} \\
& = \sum_{i \in [n]} I(Y_i:XA)_{\sigma^A} + I(X_i:BY)_{\sigma^B}
\\ & \ge n  SIC(G),
\end{align*}
where the first inequality comes from the subadditivity of the quantum conditional entropy and the last inequality comes from Lemma \ref{Lemma:SIC_i}.
Since this holds for any state $\rho$ satisfying $\omega^*(G^n|\rho) = 1$, we conclude that $SIC(G^n) \ge n SIC(G)$.

We can also notice that $SIC(G^n) \le n SIC(G)$. Indeed, consider a state $\rho$ such that $\omega^*(G|\rho) = 1$. We have $\omega^*(G^n|\rho^{\otimes n}) = 1$. Moreover, $SIC(\rho^{\otimes n}) = n SIC(\rho)$. From there, we have $SIC(G^n) \le n SIC(G)$. We conclude that $SIC(G^n) = n SIC(G)$.
\qed \end{proof}
\section{Proof of Theorems~\ref{Theorem:SingleGamePerfect} and \ref{Theorem:SingleGame}} \label{Appendix:ProofSingle}
The organisation and an overview of the proof can be found in Section~\ref{Section:Single}.

\subsection{First inequality}
We will show this inequality for any input distribution. Let $\rho = \sum_{x,y \in [k]} p_{xy} \altketbra{x}_\spa{X} \otimes \altketbra{\phi_{xy}}_{\spa{AB}} \otimes \altketbra{y}_\spa{Y}$. As in Section~\ref{Section:SuperposedDefs}, we define $\ket{L^A_y},\ket{L^B_x},\sigma^A,\sigma^B$. Let $\rho^A_y = \Tr_{\spa{B}} \ketbra{L^A_y}{L^A_y}$ and $\rho^B_x = \Tr_{\spa{A}} \ketbra{L^B_x}{L^B_x}$. Intuitively, $\rho^A_y$ (resp.  $\rho^B_x$) corresponds to the input-superposed state that Alice (resp. Bob) has, conditioned on Bob getting $y$ (resp. Alice getting $x$). We prove the following.
\begin{proposition}\label{Proposition:SingleFirst}
$SIC(\rho) \ge \frac{1}{4\ln(2)} (1 - \sum_{y,y'} p_{\cdot y} p_{\cdot y'} F^2(\rho^A_y,\rho^A_{y'}) + 1 - \sum_{x, x'} p_{x\cdot}p_{x'\cdot}F^2(\rho^B_x,\rho^B_{x'})).$
\end{proposition}
\begin{proof}
Let $\xi^A = \Tr_{B} (\sigma^A)$ and  $\xi^B = \Tr_{A} (\sigma^B)$. This means that $\xi^A =  \sum_y p_{\cdot y} \rho^A_y \otimes \altketbra{y}$ and $\xi^B =  \sum_x p_{x\cdot}\altketbra{x} \otimes \rho^B_x$. We have $SIC(\rho) = I(XA:Y)_{\xi^A} + I(X:BY)_{\xi^B}$. Using Claim~\ref{Claim:DC}, we get 
\begin{align}\label{Equation:Partial1}
SIC(\rho) \ge \frac{2}{\ln(2)} (1 - F(\xi^A,\xi^A_{\spa{XA}} \otimes \xi^A_Y) + 
1 - F(\xi^B,\xi^B_{\spa X} \otimes \xi^B_{\spa{BY}})). \end{align}
where $\xi^A_{\spa{XA}} = \sum_{y} p_{\cdot y}  \rho_{y}^A$ and $\xi^A_{\spa{Y}}  = \sum_{y} p_{\cdot y} \altketbra{y}$.
Next, using Proposition~\ref{Prop:FidelityLast}, we have
$F(\xi^A,\xi^A_\spa{XA} \otimes \xi^A_{\spa{Y}}) = \sum_{y}p_{\cdot y} F(\rho^A_y,\xi^A_{\spa{XA}}).$
From there, we have:
\begin{align*}
1 - F(\xi^A,\xi^A_\spa{XA} \otimes \xi^A_{\spa{Y}}) & = 1 - \sum_{y\in [k]} p_{\cdot y} F(\rho^A_y,\xi^A_{\spa{XA}}) \\ 
& = \frac{1}{2}(1 - \sum_{y\in [k]} p_{\cdot y} F(\rho^A_y,\xi^A_{\spa{XA}}) + 1 - \sum_{y'\in [k]} p_{\cdot y'} F(\rho^A_{y'},\xi^A_{\spa{XA}})) \\
& = \frac{1}{2} \sum_{y,y' \in [k]} p_{\cdot y}p_{\cdot y'}[1 - F(\rho^A_y,\xi^A_{\spa{XA}}) + 1 - F(\rho^A_{y'},\xi^A_{\spa{XA}})] \\
& \ge \frac{1}{4} \sum_{y,y' \in [k]} p_{\cdot y}p_{\cdot y'}(1 - F(\rho^A_{y},\rho^A_{y'})) & \mbox{ using Proposition}~\ref{Prop:FidelityInequality2} \\
& \ge \frac{1}{8} \sum_{y,y' \in [k]} p_{\cdot y}p_{\cdot y'}(1 - F^2(\rho^A_{y},\rho^A_{y'})) 
\end{align*}
Similarly, we can show that $1 - F(\xi^B,\xi^B_\spa{X} \otimes \xi^B_{\spa{BY}}) \ge \frac{1}{8} \sum_{x,x' \in [k]} p_{x \cdot}p_{x' \cdot}(1 - F^2(\rho^B_{x},\rho^B_{x'}))$. Combining these with Eq.~\ref{Equation:Partial1}, we conclude that 
\[
SIC(\rho) \ge \frac{1}{4\ln(2)} (1 - \sum_{y,y'} p_{\cdot y} p_{\cdot y'} F^2(\rho^A_y,\rho^A_{y'}) + 1 - \sum_{x, x'} p_{x\cdot}p_{x'\cdot}F^2(\rho^B_x,\rho^B_{x'})).
\]
\qed \end{proof}


\subsection{Second inequality}
Let $\rho = \frac{1}{k^2} \sum_{x,y \in [k]} \altketbra{x}_\spa{X} \otimes \altketbra{\phi_{xy}}_{\spa{AB}} \otimes \altketbra{y}_\spa{Y}$. As in Section~\ref{Section:SuperposedDefs}, we define $\ket{L^A_y},\ket{L^B_x},\sigma^A,\sigma^B$. Let $\rho^A_y = \Tr_{\spa{B}} \ketbra{L^A_y}{L^A_y}$ and $\rho^B_x = \Tr_{\spa{A}} \ketbra{L^B_x}{L^B_x}$.
We define:
\begin{align*}
\eps^A & = 1 - \sum_{y,y'} \frac{1}{k^2} F^2(\rho^A_y,\rho^A_{y'}) = 1 - \E_{y,y'}[F^2(\rho^A_y,\rho^A_{y'})] \\
\eps^B & = 1 - \sum_{x,x'} \frac{1}{k^2} F^2(\rho^B_x,\rho^B_{x'}) = 1 - \E_{x,x'}[F^2(\rho^A_x,\rho^A_{x'})].
\end{align*}
The expectations will always be taken over the uniform distribution. We first show the following lemma.
\begin{lemma}\label{Lemma:Messy}
There exist $i,j \in [k]$ as well as unitaries $\{U_x\}_x $ and $\{V_y\}_y$ acting respectively on $\spa{A}$ and $\spa{B}$ such that if we define $\ket{\Omega_{xy}} = (U_x \otimes V_y)\ket{\phi_{xy}}$, we have:
\begin{align*}
\E_{xy}[ |\braket{\Omega_{xy}}{\Omega_{xj}}|^2 ]& \ge \E_{y,y'}[F^2(\rho^A_y,\rho^A_{y'})] = 1 - \eps^A \\
\E_{xy}[ |\braket{\Omega_{xy}}{\Omega_{iy}}|^2 ]& \ge \E_{x,x'}[F^2(\rho^B_x,\rho^B_{x'})] = 1 - \eps^B. 
\end{align*}
\end{lemma}
\begin{proof}
Let $j \in [k]$ that maximizes $\E_{y'} [ F^2(\rho^A_j,\rho^A_{y'})]$. We have 
\begin{align}\label{EQ1}
\E_{y'} [ F^2(\rho^A_j,\rho^A_{y'})] \ge \E_{y,y'} [ F^2(\rho^A_y,\rho^A_{y'})] \ge 1 - \eps^A.
\end{align}
For each $y$, consider the unitary $U_y$ acting on $\spa{B}$ such that $|\bra{L^A_j}(I_{\spa{XA}} \otimes U_y) \ket{L^A_y}| = F(\rho^A_j,\rho^A_y)$. Such a unitary exists by Uhlmann's theorem. We also choose $U_j = I_{\spa{B}}$.
Since $\ket{L^A_j} = \frac{1}{\sqrt{k}} \sum_x \ket{x} \ket{\phi_{xj}}$ and $(I_{\spa{XA}} \otimes U_y)$ acts only on space $\spa{B}$, we can write $(I_{\spa{XA}} \otimes U_y) \ket{L^A_y} = \frac{1}{\sqrt{k}} \sum_{x} \ket{x} \ket{\xi_{xy}}$ for some $\ket{\xi_{xy}}$. 
Therefore, we have:
\begin{align*}
F(\rho^A_j,\rho^A_y) = |\bra{L^A_j}(I_{\spa{XA}} \otimes U_y) \ket{L^A_y}| = |\frac{1}{k} \sum_x \braket{\xi_{xy}}{\phi_{xj}}|
= |\E_x [\braket{\xi_{xy}}{\phi_{xj}}]| \le \E_x [|\braket{\xi_{xy}}{\phi_{xj}}|].
\end{align*}
Since we took $U_j = I_{\spa{B}}$, we have $\ket{\xi_{xj}} = \ket{\phi_{xj}}$ for all $x$. We can hence rewrite for all $y$
\begin{align}\label{EQ2}
F(\rho^A_j,\rho^A_y) \le \E_x [\braket{\xi_{xy}}{\xi_{xj}}].
\end{align}

We now analyze Bob's side of the state similarly.
Let $\ket{M^B_{x}} = \sum_y \frac{1}{\sqrt{k}}\ket{\xi_{xy}} \ket{y}$. We have $\ket{M^B_x} = (\sum_y I_{\spa{A}} \otimes U_y \otimes \ketbra{y}{y} ) \ket{L^B_x}$. Let $\nu^B_x = Tr_{\spa{A}} \ketbra{M^B_x}{M^B_x}$. We have $\nu^B_x = (\sum_y U_y^\dagger U_y \otimes \ketbra{y}{y} ) \cdot \rho^B_x$. Hence for all $x,x'$, we have 

\begin{align}\label{EQ3} F(\nu^B_x,\nu^B_{x'}) = F(\rho^B_x,\rho^B_{x'}).\end{align} 

Let $i \in [k]$ such that $\E_{x'} [ F^2(\nu^B_i,\nu^B_{x'})] $ is maximal. We have 
\begin{align}\label{EQ4} \E_{x'} [ F^2(\nu^B_i,\nu^B_{x'})] \ge 1 - \eps^B.\end{align}

For each $x$, consider the unitary $V_x$ acting on $\spa{A}$ such that $|\bra{M^B_i}(V_x \otimes I_{\spa{BY}}) \ket{M^B_x}| = F(\nu^B_i,\nu^B_x)$. Such a unitary exists by Uhlmann's theorem. We take $V_i = I_\spa{A}$.
Since $\ket{M^B_x} = \frac{1}{\sqrt{k}} \sum_y \ket{\xi_{xy}} \ket{y}$ and $(V_x \otimes I_{\spa{BY}})$ acts only on space $\spa{A}$, we can write $(V_x \otimes I_{\spa{BY}}) \ket{M^B_x} = \frac{1}{\sqrt{k}} \sum_{y} \ket{\Omega_{xy}} \ket{y}$ for some $\ket{\Omega_{xy}}$. Therefore, we have:
\begin{align*}
F(\nu^B_i,\nu^B_x) = |\bra{M^B_i}(V_x \otimes I_{\spa{BY}}) \ket{M^B_x}| = |\frac{1}{k} \sum_y \braket{\xi_{iy}}{\Omega_{xy}}|
= |\E_y [\braket{\xi_{iy}}{\Omega_{xy}}]| \le \E_y [|\braket{\xi_{iy}}{\Omega_{xy}}|].
\end{align*}

Using $F(\nu^B_i,\nu^B_i) = 1$, we have $\ket{\xi_{iy}} = \ket{\Omega_{iy}}$ for all $y$. Using Eq.~\ref{EQ3}, we can hence rewrite for all $x$:
\begin{align} \label{EQ5}
F(\rho^B_i,\rho^B_x) = F(\nu^B_i,\nu^B_x) \le \E_y [|\braket{\Omega_{iy}}{\Omega_{xy}}| ]. 
\end{align}
Note finally that for all $x$, $(V_x \otimes I_{\spa{B}})(\ket{\xi_{xy}}) = \ket{\Omega_{xy}}$ hence we have for all $x$ and for all $y$ $\braket{\Omega_{xy}}{\Omega_{xy'}} = \braket{\xi_{xy}}{\xi_{xy'}}$. Using Eq.~\ref{EQ2}, we have

\begin{align}\label{EQ6}
F(\rho^A_j,\rho^A_y) = \E_x [ |\braket{\xi_{xy}}{\xi_{xj}}| ] = \E_x[ |\braket{\Omega_{xy}}{\Omega_{xj}}|].
\end{align}

Equations~\ref{EQ5} and \ref{EQ6} give
\begin{align*}
F^2(\rho_j^A,\rho_y^A)  = \E_x [ |\braket{\Omega_{xy}}{\Omega_{xj}}|]^2 \le \E_x [|\braket{\Omega_{xy}}{\Omega_{xj}}|^2] \\
F^2(\rho_i^B,\rho_x^B)  = \E_y [|\braket{\Omega_{xy}}{\Omega_{iy}}|]^2 \le \E_x[ |\braket{\Omega_{xy}}{\Omega_{iy}}|^2] .
\end{align*}
Combining this with equations~\ref{EQ1} and \ref{EQ4}, we conclude
\begin{align*}
1 - \eps^A \le \E_y [ F^2(\rho_j^A,\rho_y^A) ] \le \E_{xy} [ |\braket{\Omega_{xy}}{\Omega_{xj}}|^2 ]\\
1 - \eps^B \le \E_x [F^2(\rho_i^B,\rho_x^B)]  \le \E_{xy} [ |\braket{\Omega_{xy}}{\Omega_{iy}}|^2 ].
\end{align*}
\qed \end{proof}
We can now prove the main proposition of this section.
\begin{proposition}\label{Proposition:SingleMiddle}
For any state $\rho = \frac{1}{k^2}\sum_{x,y \in [k]} \altketbra{x} \otimes \altketbra{\phi_{xy}} \otimes \altketbra{y}$, there exist unitaries $\{U_x\}_x$ and $\{V_y\}_y$ such that 
\[
\eps^A + \eps^B \ge \frac{1}{8}(1 - \max_{\ket{\Omega}}\sum_{x,y \in [k]} \frac{1}{k^2} |\triple{\Omega}{U_x \otimes V_y}{\phi_{xy}}|^2),
\]
where
$
\eps^A  = 1 - \sum_{y,y'} \frac{1}{k^2} F^2(\rho^A_y,\rho^A_{y'})$ and
$\eps^B = 1 - \sum_{x,x'} \frac{1}{k^2} F^2(\rho^B_x,\rho^B_{x'}).
$
\end{proposition}
\begin{proof}
Fix $\rho = \frac{1}{k^2}\sum_{x,y \in [k]} \altketbra{x} \otimes \altketbra{\phi_{xy}} \otimes \altketbra{y}$. 
Using Lemma~\ref{Lemma:Messy}, let $\{U_x\}_x, \{V_y\}_{y}, i,j$ such that 
\begin{align*}
\E_{xy}[ |\braket{\Omega_{xy}}{\Omega_{xj}}|^2 ]& \ge 1 - \eps^A \\
\E_{xy}[ |\braket{\Omega_{xy}}{\Omega_{iy}}|^2 ]& \ge 1 - \eps^B,
\end{align*}
with $\ket{\Omega_{xy}} = U_x \otimes V_y \ket{\phi_{xy}}$. Using Proposition~{\ref{Proposition:Trig}}, we have
\begin{align*}
\E_{x,y,y'} [ |\braket{\Omega_{xy}}{\Omega_{xy'}}|] & \ge 
\E_{x,y,y'} [|\braket{\Omega_{xy}}{\Omega_{xj}}|^2 + |\braket{\Omega_{xj}}{\Omega_{xy'}}|^2] - 1 \\
& \ge 1 - \eps^A + 1 - \eps^A - 1 = 1 - 2\eps^A.
\end{align*} 
It follows that
\begin{align*}
\E_{x,y,y'} [ |\braket{\Omega_{xy}}{\Omega_{xy'}}|^2 ]& \ge
\E_{x,y,y'}[|\braket{\Omega_{xy}}{\Omega_{xy'}}|]^2 
 \ge (1 - 2\eps^A)^2 \ge 1 - 4\eps^A.
\end{align*}
Similarly, we get
$ 
\E_{x,x',y} [|\braket{\Omega_{xy}}{\Omega_{x'y}}|^2 ]\ge 1 - 4 \eps^B.
$
Using Proposition~{\ref{Proposition:Trig}} again, we have
\begin{align*}
\E_{x,x',y,y'} [ |\braket{\Omega_{xy}}{\Omega_{x'y'}}| ]& \ge
\E_{x,x',y,y'} [ |\braket{\Omega_{xy}}{\Omega_{x'y}}|^2 + 
|\braket{\Omega_{x'y}}{\Omega_{x'y'}}|^2 ] - 1 \\
& \ge 1 - 4\eps^A + 1 - 4\eps^B - 1 = 1 - 4(\eps^A + \eps^B).
\end{align*}
This gives us 
\begin{align*}
\E_{x,x',y,y'} [|\braket{\Omega_{xy}}{\Omega_{x'y'}}|^2] \ge
\E_{x,x',y,y'} [ |\braket{\Omega_{xy}}{\Omega_{x'y'}}]^2 \ge (1 - 4\eps^A - 4\eps^B)^2 \ge 1 - 8\eps^A - 8\eps^B.
\end{align*}
Using
\begin{align*}
\E_{x,y,x',y'} [|\braket{\Omega_{xy}}{\Omega_{x'y'}}|^2] \le 
\max_{x'y'} \E_{x,y} [|\braket{\Omega_{xy}}{\Omega_{x'y'}}|^2]
\le \max_{\ket{\Omega}} \E_{x,y} [|\braket{\Omega_{xy}}{\Omega}|^2],
\end{align*}
we have
$$
\max_{\ket{\Omega}} \E_{x,y} [ |\braket{\Omega}{\Omega_{xy}}|^2) ] \ge  \E_{x,x',y,y'}[|\braket{\Omega_{xy}}{\Omega_{x'y'}}|^2] \ge 1 - 8\eps^A - 8\eps^B,
$$
hence
\[
\eps^A + \eps^B \ge \frac{1}{8} (1 - \max_{\ket{\Omega}} (\E_{x,y} [ |\braket{\Omega}{\Omega_{xy}}|^2 ])) = \frac{1}{8} (1 - \max_{\ket{\Omega}} (\E_{x,y} [|\triple{\Omega}{U_x \otimes V_y}{\phi_{xy}}|^2 ])).
\]
\qed \end{proof}


\subsection{Last inequality}
\begin{proposition}\label{Prop:MaxEigenvalueOmega}
Consider a game $G = (I,O,V,p)$ and a state $$\rho = \sum_{x,y \in I} p_{xy} \altketbra{x} \otimes \altketbra{\phi_{xy}} \otimes \altketbra{y}$$ such that $\omega^*(G|\rho) = 1$. We have that $\max_{\ket{\Omega}} \sum_{x,y \in I} p_{xy} |\braket{\Omega}{\phi_{xy}}|^2 \le \omega^*(G)$.
\end{proposition}
\begin{proof}
Consider strategies $\{A^x_a\}_{x \in I, a \in O}$ and  $\{B^y_b\}_{y \in I, b \in O}$ such that
\begin{align*} 
\sum_{x,y,a,b} p_{xy} V(a,b|x,y) \triple{\phi_{xy}}{A^x_a \otimes B^y_b}{\phi_{xy}} = 1.
\end{align*}
Let $\ket{\Omega_0}$ that maximizes $\sum_{x,y \in I} p_{xy} |\braket{\Omega_0}{\phi_{xy}}|^2$. For any $x,y$, since \\ $\sum_{ab} V(a,b|x,y) \triple{\phi_{xy}}{A^x_a \otimes B^y_b}{\phi_{xy}} = 1$, we have:
\begin{align*}
\sum_{a,b} V(a,b|x,y) \triple{\Omega_0}{A^x_a \otimes B^y_b}{\Omega_0} \ge |\braket{\Omega_0}{\phi_{xy}}|^2.
\end{align*}
From there, we have:
\begin{align*}
\omega^*(G) & \ge \sum_{xyab} p_{xy} V(a,b|x,y) \triple{\Omega_0}{A^x_a \otimes B^y_b}{\Omega_0} \\ & \ge \sum_{xy} p_{xy} |\braket{\Omega_0}{\phi_{xy}}|^2 = \max_{\ket{\Omega}} \sum_{xy} p_{xy} |\braket{\Omega}{\phi_{xy}}|^2.
\end{align*}
\qed \end{proof}
This proposition has a useful corollary:
\begin{corollary}\label{Corollary2perfect}
Consider a game $G = (I,O,V,p)$ and a state $$\rho = \sum_{x,y \in I} p_{xy} \altketbra{x}_\spa{X} \otimes \altketbra{\phi_{xy}}_\spa{AB} \otimes \altketbra{y}_\spa{Y}$$ such that $\omega^*(G|\rho) = 1$. We have $$\max_{\ket{\Omega},\{U_x\},\{V_y\}} \sum_{x,y \in I} p_{xy} |\bra{\Omega} (U_x \otimes V_y)\ket{\phi_{xy}}|^2 \le \omega^*(G),$$ for unitaries $\{U_x\}_x$ and $\{V_y\}_y$ acting respectively on $\spa{A}$ and $\spa{B}$.
\end{corollary}
\begin{proof}
Let $\{U_x\}_x$,$\{V_y\}_y$ that maximize $\max_{\ket{\Omega}} \sum_{x,y \in I} p_{xy} |\bra{\Omega} (U_x \otimes V_y)\ket{\phi_{xy}}|^2$. Let $\ket{\psi_{xy}} = U_x \otimes V_y \ket{\phi_{xy}}$. Let $\eta = \sum_{xy} p_{xy} \altketbra{x} \otimes \altketbra{\psi_{xy}} \otimes \altketbra{y}$. Since Alice and Bob can go from $\eta$ to $\rho$  by applying respectively $U_x^\dagger$ and $V_y^\dagger$, we conclude that $\omega^*(G|\eta) = \omega^*(G|\rho) = 1$. Using Proposition~\ref{Prop:MaxEigenvalueOmega}, we have $\max_{\ket{\Omega}} \sum_{x,y \in I} p_{xy} |\bra{\Omega} (U_x \otimes V_y)\ket{\phi_{xy}}|^2 = \max_{\ket{\Omega}} \sum_{x,y \in I} p_{xy} |\braket{\Omega}{\psi_{xy}}|^2  \le \omega^*(G)$.
\qed \end{proof}

We now prove a similar statement in the case $\omega^*(G|\rho) < 1$.
\begin{proposition}\label{Prop:MaxEigenvalueOmegaImperfect}
Consider a game $G = (I,O,V,p)$ and a state $$\rho = \sum_{x,y \in I} p_{xy} \altketbra{x} \otimes \altketbra{\phi_{xy}} \otimes \altketbra{y}.$$ If $\omega^*(G|\rho) \ge 1 - \gamma$ and $\max_{\ket{\Omega}} \sum_{x,y \in I} p_{xy} |\braket{\Omega}{\phi_{xy}}|^2 \ge 1 - \gamma'$, then $$\omega^*(G) \ge 1 - 2(\gamma + \gamma').$$
\end{proposition}
\begin{proof}
Consider strategies $\{A^x_a\}_{x \in I, a \in O}$ and  $\{B^y_b\}_{y \in I, b \in O}$ such that
\begin{align*} 
\sum_{x,y,a,b} p_{xy} V(a,b|x,y) \triple{\phi_{xy}}{A^x_a \otimes B^y_b}{\phi_{xy}} = 1 - \gamma.
\end{align*}
Let $M^{xy} = \sum_{a,b} V(a,b|x,y) A^x_a \otimes B^y_b$ and let $\ket{C_{xy}} = \frac{M^{xy} \ket{\phi_{xy}}}{||M^{xy} \ket{\phi_{xy}}||}$. We have $tr(M^{xy} \altketbra{\phi_{xy}}) = |\braket{C_{xy}}{\phi_{xy}}|^2$. Let $q_{xy} = |\braket{C_{xy}}{\phi_{xy}}|^2$. This gives us immediately
\begin{align*}
\sum_{x,y} p_{xy} q_{xy} = 1 - \gamma.
\end{align*}

Let $\ket{\Omega}$ such that $\sum_{x,y \in I} p_{xy} |\braket{\Omega}{\phi_{xy}}|^2 \ge 1 - \gamma'$.  Also, let $r_{xy} = |\braket{\Omega}{\phi_{xy}}|^2$ and $s_{xy} = |\braket{\Omega}{C_{xy}}|^2$. We have that
\begin{align*}
\sum_{xy} p_{xy} r_{xy} \ge 1 - \gamma' ,
\end{align*}
as well as 
\begin{align*}
\omega^*(G) \ge  \sum_{xy} p_{xy} tr(M^{xy} \altketbra{\Omega}) \ge \sum_{xy} p_{xy} |\braket{\Omega}{C_{xy}}|^2 = \sum_{xy} p_{xy} s_{xy}.
\end{align*}
Using Proposition~\ref{Proposition:Trig}, we have that for all $x,y$
$
s_{xy} \ge (q_{xy} + r_{xy} - 1)^2$. Let $m_{xy} = 1 - q_{xy} + 1 - r_{xy}$. We have by defintion that $\sum_{xy} p_{xy} m_{xy} \le \gamma + \gamma'$. Moreover, we have:
\begin{align*}
\sum_{xy} p_{xy} s_{xy} & \ge \sum_{xy} p_{xy} (q_{xy} + r_{xy} - 1)^2 \\
& = \sum_{xy} p_{xy} (1 - m_{xy})^2 \\
& \ge \sum_{xy} p_{xy} (1 - 2 m_{xy})
\ge 1 - 2(\gamma + \gamma').
\end{align*}
We conclude that $\omega^*(G) \ge \sum_{xy} p_{xy} s_{xy} \ge 1 - 2(\gamma + \gamma')$.
\qed \end{proof}
We derive two corollaries from this proposition.
\begin{corollary}
Consider a game $G = (I,O,V,p)$ and a state $$\rho = \sum_{x,y \in I} p_{xy} \altketbra{x}_\spa{X} \otimes \altketbra{\phi_{xy}}_\spa{AB} \otimes \altketbra{y}_\spa{Y}.$$ If $\omega^*(G|\rho) \ge 1 - \gamma$ and $$\max_{\ket{\Omega},\{U_x\},\{V_y\}} \sum_{x,y \in I} p_{xy} |\bra{\Omega} (U_x \otimes V_y)\ket{\phi_{xy}}|^2 \ge 1 - \gamma',$$ for unitaries $\{U_x\}_x$ and $\{V_y\}_y$ acting respectively on $\spa{A}$ and $\spa{B}$, then $$\omega^*(G) \ge 1 - 2(\gamma + \gamma').$$
\end{corollary}
\begin{proof}
Let $\{U_x\}$,$\{V_y\}$ such that $\max_{\ket{\Omega}} \sum_{x,y \in I} p_{xy} |\bra{\Omega} (U_x \otimes V_y)\ket{\phi_{xy}}|^2 = 1 - \gamma'$. Let $\ket{\psi_{xy}} = U_x \otimes V_y \ket{\phi_{xy}}$. Let $\eta = \sum_{xy} p_{xy} \altketbra{x} \otimes \altketbra{\psi_{xy}} \otimes \altketbra{y}$. Since Alice and Bob can go from $\eta$ to $\rho$  by applying respectively $U_x^\dagger$ and $V_y^\dagger$, we conclude that $\omega^*(G|\eta) = \omega^*(G|\rho) \ge 1 - \gamma$. Using Proposition~\ref{Prop:MaxEigenvalueOmegaImperfect}, we conclude that $\omega^*(G) \ge 1 - 2(\gamma + \gamma')$.
\qed \end{proof}
Taking a counterpostitive of the above Corollary we get the following
\begin{corollary}\label{Corollary2}
Consider a game $G = (I,O,V,p)$ and a state $$\rho = \sum_{x,y \in I} p_{xy} \altketbra{x}_\spa{X} \otimes \altketbra{\phi_{xy}}_\spa{AB} \otimes \altketbra{y}_\spa{Y}.$$ If $\omega^*(G|\rho) \ge 1 - \gamma$ and $\omega^*(G) \le 1 - \eps$, then $$\max_{\ket{\Omega},\{U_x\},\{V_y\}} \sum_{x,y \in I} p_{xy} |\bra{\Omega} (U_x \otimes V_y)\ket{\phi_{xy}}|^2 \le 1 - (\eps/2 - \gamma),$$ for unitaries $\{U_x\}_x$ and $\{V_y\}_y$ acting respectively on $\spa{A}$ and $\spa{B}$.
\end{corollary}


\subsection{Putting it together}
We can now show our theorems
\addtocounter{theorem}{0}
\begin{theorem}\label{Theorem:SingleGamePerfect}
For any game $G$ with a uniform input distribution, we have  
$SIC(G) \ge \frac{1 - \omega^*(G)}{32\ln(2)}$.
\end{theorem}
\begin{proof}
Consider a game $G = (I = [k],O,V,\mbox{Unif.})$ and $\rho = \frac{1}{k^2} \sum_{x,y} \altketbra{x} \otimes \altketbra{\phi_{xy}} \otimes \altketbra{y}$ such that $\omega^*(G|\rho) = 1$. Using Proposition~\ref{Proposition:SingleFirst} and Proposition~\ref{Proposition:SingleMiddle}, take $\{U_x\}_x$ and $\{V_y\}_y$ such that 
\begin{align*}
SIC(\rho) \ge \frac{1}{32\ln(2)} (1 - \max_{\ket{\Omega}} \sum_{xy} \frac{1}{k^2} | \triple{\Omega}{(U_x \otimes V_y)}{\phi_{xy}}|^2).
\end{align*}
Using Corollary~\ref{Corollary2perfect}, we have
\begin{align*}
\max_{\ket{\Omega}} \sum_{xy \in [k]} \frac{1}{k^2} |\braket{\Omega}{(U_x \otimes V_y)\phi_{xy}}|^2 \le \omega^*(G).
\end{align*}
From there, we have $SIC(\rho) \ge \frac{1 - \omega^*(G)}{32\ln(2)}$. Since this holds for any $\rho$ satisfying $\omega^*(G|\rho) = 1$, we can conclude that $SIC(G) \ge \frac{1 - \omega^*(G)}{32\ln(2)}$.
\qed \end{proof}

We now proceed to prove a similar result for the case where $\omega^*(G|\rho) < 1$.
\begin{proposition}
For any game $G$ with a uniform input distribution, and any state $\rho$ such that $\omega^*(G|\rho) = 1 - \gamma$, we have  
$SIC(\rho) \ge \frac{1}{32\ln(2)} (\frac{\eps}{2} - \gamma)$
where $\eps = 1 - \omega^*(G)$.
\end{proposition}
\begin{proof}
The proof will be similar to the previous one. Consider a game $G = (I = [k],O,V,\mbox{Unif.})$ and $\rho = \frac{1}{k^2} \sum_{x,y} \altketbra{x} \otimes \altketbra{\phi_{xy}} \otimes \altketbra{y}$ such that $\omega^*(G|\rho) = 1 - \gamma$. Using Proposition~\ref{Proposition:SingleFirst} and Proposition~\ref{Proposition:SingleMiddle}, take $\{U_x\}$ and $\{V_y\}$ such that 
\begin{align*}
SIC(\rho) \ge \frac{1}{32\ln(2)} (1 - \max_{\ket{\Omega}}\sum_{xy} \frac{1}{k^2}|\braket{\Omega}{(U_x \otimes V_y)\phi_{xy}}|^2).
\end{align*}
Using Corollary~\ref{Corollary2}, we have that
\begin{align*}
1 - \max_{\ket{\Omega}} \sum_{xy} \frac{1}{k^2}|\braket{\Omega}{(U_x \otimes V_y)\phi_{xy}}|^2 \ge \frac{\eps}{2} - \gamma,
\end{align*}
where $\eps = 1 - \omega^*(G)$.
From there, we have $SIC(\rho) \ge \frac{1}{32\ln(2)}(\frac{\eps}{2} - \gamma)$.  Since this holds for any $\rho$ satisfying $\omega^*(G|\rho) = 1$, we can conclude that $SIC(G) \ge \frac{1 - \omega^*(G)}{32\ln(2)}$.
\qed \end{proof}

Our last extension is the following theorem, which is the one we will use for parallel repetition.

\begin{theorem}\label{Theorem:SingleGame}
There exists a constant $c_0>0$ such that for any game $G = (I = [k],O,V,\mbox{Unif.})$ satisfying $\omega^*(G) = 1 - \eps$, for any game $G' = (I = [k],O,V,p)$ satisfying $\frac{1}{2} \sum_{x,y} |p_{xy} - \frac{1}{k^2}| \le c_0{\eps}$ and any state $\rho = \sum_{xy} p_{xy} \altketbra{x} \otimes \altketbra{\phi_{xy}} \otimes \altketbra{y}$ such that $\omega^*(G'|\rho) \ge 1 - \frac{\eps}{4}$, we have that $SIC(\rho) = \Omega(\eps)$.
\end{theorem}
\begin{proof}
Fix any $G,G',\rho$. We also fix a small constant $c_0$ that will be specified later in the proof. Let $\rho(U) = \frac{1}{k^2} \sum_{xy} \altketbra{x} \otimes \altketbra{\phi_{xy}} \otimes \altketbra{y}$. 

Let $\sigma^A,\sigma^B$ the superposed states of $\rho$. As in Proposition~\ref{Proposition:SingleFirst}, we define $\xi^B = Tr_{\spa{A}} (\sigma^B)$. This means that $\xi^B = \sum_x p_{x \cdot} \altketbra{x} \otimes \rho^B_x$ for some $\rho^B_x$. Let also $\xi^B_\spa{X} = \Tr_{BY} (\xi^B)$ and $\xi^B_{\spa{BY}} = \Tr_{X} (\xi^B)$. 

Similarly, let $\sigma^A(U),\sigma^B(U)$ the superposed states of $\rho(U)$ and let $\xi^B(U) = Tr_{\spa{A}} (\sigma^B(U))$. This means that $\xi^B(U) = \frac{1}{k} \sum_x \altketbra{x} \otimes \rho^B_x(U)$ for some $\rho^B_x(U)$. Let also $\xi^B_\spa{X}(U) = \Tr_{BY} (\xi^B(U))$ and $\xi^B_{\spa{BY}}(U) = \Tr_{X} (\xi^B(U))$.

We want to upper bound $SIC(\rho) = I(Y:XA)_{\sigma^A} + I(X:BY)_{\sigma^B}$. Let $\delta = \frac{1}{2}\sum_{x,y} |p_{xy} - \frac{1}{k^2}| \le c_0\eps$. We proceed as in Proposition~\ref{Proposition:SingleFirst}. Using Claim~\ref{Claim:DC}, we have
$
I(X:BY)_{\sigma^B} \ge \frac{2}{\ln(2)} (1 - F(\xi^B,\xi^B_\spa{X} \otimes \xi^B_{\spa{BY}})) $. Notice that $\Delta(\sigma^B,\sigma^B(U)) \le \delta$ which implies $
\Delta(\xi^B(U),\xi^B) \le \delta$ ; 
$\Delta(\xi^B_\spa{X}(U),\xi^B_\spa{X}) \le \delta$ and $\Delta(\xi^B_\spa{BY}(U),\xi^B_\spa{BY}) \le \delta
$. 
The two last inequalities give us $\Delta (\xi^B_\spa{X}(U) \otimes 
\xi^B_\spa{BY}(U), \xi^B_\spa{X} \otimes \xi^B_\spa{BY}) \le 2\delta$.
From there, by using Claim~\ref{Claim:DC} and Propositions~\ref{FuchsVdG} and~\ref{Prop:FidelityInequality2}, we have:
\begin{align*}
I(X:BY)_{\sigma^B} = I(X:BY)_{\xi^B} & \ge \frac{2}{\ln(2)} (1 - F(\xi^B,\xi^B_\spa{X} \otimes \xi^B_\spa{BY})) \\
& \ge \frac{2}{\ln(2)} (\frac{1}{2} (1 - F(\xi^B(U),\xi^B_\spa{X} \otimes \xi^B_{\spa{BY}})) - (1 - F(\xi^B,\xi^B(U)))) \\
& \ge \frac{2}{\ln(2)} (\frac{1}{2} (1 - F(\xi^B(U),\xi^B_\spa{X} \otimes \xi^B_{\spa{BY}})) - \delta ).
\end{align*}
Then, we have:
\begin{align*}
1 - F(\xi^B(U),\xi^B_\spa{X} \otimes \xi^B_{\spa{BY}}) & \ge \frac{1}{2} (1 - F(\xi^B(U),\xi^B_\spa{X}(U) \otimes \xi^B_{\spa{BY}}(U))) - (1 - F(\xi^B_\spa{X}\otimes \xi^B_{\spa{BY}},\xi^B_\spa{X}(U) \otimes \xi^B_{\spa{BY}}(U))) \\
& \ge \frac{1}{2} (1 - F(\xi^B(U),\xi^B_\spa{X}(U) \otimes \xi^B_{\spa{BY}}(U))) - 2\delta,
\end{align*}
which gives us
\begin{align*}
I(X:BY)_{\sigma^B} \ge \frac{2}{\ln(2)}(\frac{1}{4} (1 - F(\xi^B(U),\xi^B_\spa{X}(U) \otimes \xi^B_{\spa{BY}}(U))) - 2\delta).
\end{align*}
Let $\eps^B = 1 - \frac{1}{k^2}\sum_{x,x'}F^2(\rho_{x}^B(U),\rho_{x'}^B(U))$. As in Proposition~\ref{Proposition:SingleFirst}, we can show that 
$$
(1 - F(\xi^B(U),\xi^B_\spa{X}(U) \otimes \xi^B_{\spa{BY}}(U))) \ge \frac{\eps^B}{8},
$$
hence $I(X:BY)_{\sigma^B} \ge \frac{2}{\ln(2)} (\frac{\eps^B}{32} -2\delta)$. Similarly, if we define $\eps^A = 1 - \frac{1}{k^2}\sum_{y,y'}F^2(\rho_{y}^A(U),\rho_{y'}^A(U))$  we can show that $I(Y:XA)_{\sigma^A} \ge \frac{2}{\ln(2)} (\frac{\eps^A}{32} -2\delta)$, which gives 
\begin{align*}
SIC(\rho) \ge \frac{2}{\ln(2)} \left(\frac{\eps^A + \eps^B}{32} -4\delta \right).
\end{align*}
Using Proposition~\ref{Proposition:SingleMiddle}, we have:
\begin{align*}
SIC(\rho) \ge \frac{2}{\ln(2)} \left(\frac{1}{256}\max_{\ket{\Omega},\{U_x\},\{V_y\}} \frac{1}{k^2} \sum_{x,y} |\braket{\Omega}{\phi_{xy}}|^2 -4\delta \right).
\end{align*}

We have $\omega(G) = 1 - \eps$ and $\omega(G|\rho(U)) \ge 1 - \eps/4 - \delta$. Using Corollary~\ref{Corollary2}, we have:
$$
\max_{\ket{\Omega},\{U_x\},\{V_y\}} \frac{1}{k^2} \sum_{x,y} |\braket{\Omega}{\phi_{xy}}|^2 \le 1 - (\eps/2 - \eps/4 - \delta) = 1 - \eps/4 + \delta.
$$
From there, we conclude:
\begin{align*}
SIC(\rho) \ge \frac{2}{\ln(2)} (\frac{1}{256}(\eps/4 - \delta) -4\delta).
\end{align*}
By taking $c_0 = \frac{1}{8092}$, which implies $\delta \le \frac{\eps}{8092}$, we obtain $SIC(\rho) = \Omega(\eps)$.
\qed \end{proof}

\section{Proof of Theorem \ref{Theorem:UpperBound}} \label{Appendix:UpperBound}

We first present the actual construction of $\xi$ and then show it has the desired properties required for Theorem~\ref{Theorem:UpperBound}.
\begin{itemize}
\item Alice and Bob perform protocol $P$ where the inputs are classical but the randomness, the message and the outputs are left in a quantum superposition. To maintain the ``classicality'' of the message sent by Alice, we ask Alice to have a quantum register which acts as a copy of the message.
\item We ask Bob to determine whether he aborts or not. The state $\xi$ will be the state Alice and Bob share conditioned on Bob not aborting.  
\item Using Proposition~\ref{Proposition:ProtocolEfficiency}, we prove that $\xi$ has the desired properties
\end{itemize}

\ \\ 

\cadre{
\begin{center}Procedure for constructing $\xi$ \end{center}
\begin{enumerate}
\item Alice and Bob pick random inputs $x,y \in_R I^n = [k^n]$. They also share a state $\sum_{r} \gamma_r \ket{r}_{\spa{R}_A} \otimes \ket{\phi}_{\spa{AB}} \otimes \ket{r}_{\spa{R}_B}$ where $\ket{\phi}$ is the same as in protocol $P$ and $r$ corresponds to the shared randomness in protocol $P$.
\item Alice and Bob perform a strategy that allows them to win $G^n$ with probability $2^{-(t+1)}$ but keep their outputs in a coherent superposition instead of measuring. They keep these outputs in registers $\spa{O}_A$ and $\spa{O}_B$. They hence share the state 
$
\rho_1 = \sum_{x,y} \frac{1}{k^{2n}} \altketbra{x}_{\spa{X}} \otimes \altketbra{\Omega^1_{xy}} \otimes \altketbra{y}_\spa{Y},
$
with 
\[
\ket{\Omega^1_{xy}} = \sum_{a,b,r} \gamma_{xyrab} \ket{a}_{\spa{O}_A} \ket{r}_{\spa{R}_A} \ket{\phi^{xy}_{ab}}_{\spa{AB}} \ket{r}_{\spa{R}_B} \ket{b}_{\spa{O}_B},\]
for some states $\ket{\phi^{xy}_{ab}}$.
\item Alice sends the message $M$ that depends on $x,a,r$ corresponding to step 3 of protocol $P$ to Bob and keeps a copy of $M$ to herself in superposition, which means that they share a state $
\rho_2 = \sum_{x,y} \frac{1}{k^{2n}} \altketbra{x} \otimes \altketbra{\Omega^2_{xy}} \otimes \altketbra{y},
$
with
\[
\ket{\Omega^2_{xy}} = \sum_{a,b,r,M} \gamma_{xyrabM} \ket{a}_{\spa{O}_A} \ket{M}_{\spa{M}_A} \ket{r}_{\spa{R}_A} \ket{\phi^{xy}_{ab}}_{AB} \ket{r}_{\spa{R}_B} \ket{M}_{\spa{M}_B} \ket{b}_{\spa{O}_B}.\]
\item Bob copies in a new register $\spa{Z}$ whether he aborts or not. This means that they share a state $
\rho_3 = \sum_{x,y} \frac{1}{k^{2n}} \altketbra{x} \otimes \altketbra{\Omega^3_{xy}} \otimes \altketbra{y},
$
with  
\begin{align*}
\ket{\Omega^3_{xy}}  = \sum_{a,b,r,M} \gamma_{xyrabM} \ket{a} \ket{M} \ket{r} \ket{\phi^{xy}_{ab}} \ket{r} \ket{M} \ket{b} \ket{NA}_\spa{Z} + \\ 
 \sum_{a,r,M} \gamma_{xyra,AB,M} \ket{a} \ket{M} \ket{r} \ket{\phi^{xy}_{a,AB}} \ket{r} \ket{m} \ket{AB} \ket{AB}_\spa{Z}.
\end{align*} 
We can write $\ket{\Omega^3_{xy}} = \sqrt{\gamma'_{xy}} \ket{Y^{NA}_{xy}}\ket{NA} + \sqrt{1 - \gamma'_{xy}} \ket{Y^{AB}_{xy}}\ket{AB},$ for some $\{\gamma'_{xy}\}_{xy}$ and states  $\{\ket{Y}^{NA}_{xy}\}_{xy}$ and $\{\ket{Y}^{AB}_{xy}\}_{xy}$.
\end{enumerate}
 Let $\rho_{-Z} = Tr_{Z} (\rho_3)$. Since the probability of Bob not aborting is $p$, we can write
\begin{align*}
\rho_{-Z} = p \cdot \rho_{NA} + (1-p) \cdot \rho_{AB},
\end{align*}
for some state $\rho_{AB}$.
$\rho_{NA}$ is of the form $\sum_{xy} q_{xy} \altketbra{x} \otimes \altketbra{Y^{NA}_{xy}} \otimes \altketbra{y}$. We choose $\xi = \rho_{NA}$. \\
} $ \ $ \\ \\
\\ In the above protocol, $\rho_2$ corresponds to the state Alice and Bob share after Step 3 of protocol $P$ except that the randomness, message and outputs are kept in a quantum superposition in the way described above.

Similarly, $\xi = \rho_{NA}$ corresponds to the state at the end of protocol $P$, conditioned on Bob not aborting. Again, the randomness, message and outputs are kept in a quantum superposition in the way described above. 

\subsection{Showing the desired properties of $\xi = \rho_{NA}$}

We now show that $\xi = \rho_{NA}$ has the desired properties of Theorem~\ref{Theorem:UpperBound}.

\paragraph{1) $H(XY)_{\xi} \ge 2n\log(k) - t - 1.$ }
\begin{proof}
$H(XY)_{\rho_{-Z}} = 2n\log(k)$. Since $Dim(XY) = k^{2n}$, this means that $H_{min}(XY)_{\rho_{-Z}} = 2n\log(k)$. We have $p \rho_{NA} \le \rho_{-Z}$ hence $ H_{min}(XY)_{\rho_{NA}} - \log(p) \ge H_{min}(XY)_{\rho_{-Z}} = 2n\log(k)$. This gives us
$H_{min}(XY)_{\rho_{NA}} \ge 2n\log(k) + log(p)$. Since $p\ge 2^{- (t+ 1)}$, we conclude that $H_{min}(XY)_{\rho_{NA}} \ge 2n\log(k) - t - 1$, hence $H(XY)_{\rho_{NA}} \ge 2n\log(k) - t - 1$.
\qed \end{proof}

\paragraph{2) $\omega^*(G'|\xi) \ge 1 - \eps/32$  \textnormal{ where } $G^n_{1 - \eps/32} = (I',O',V',\mbox{\textnormal{Unif.}})$ \textnormal{ and } $ G' = (I',O',V',q).$}
\begin{proof}
This holds by construction of $\xi$. Indeed, $\xi$ is the superposed version of the state Alice and Bob share after protocol $P$ conditionned on Bob not aborting. We know that in this case, $Pr[\textrm{Alice and Bob win } \ge (1-\eps/32)n \textrm{ games }| \textrm{ Bob does not abort} ] \ge (1-\eps/32)$. From there, we have $\omega^*(G'|\xi) \ge (1-\eps/32)$
\qed \end{proof}

\paragraph{3)  $SIC(\xi) \le \frac{32 \log(|I||O|)}{\eps} ((t+1) + |\log(\eps)| + 5) + 2t + 2.$} $ \ $ 
\begin{proof} 
We upper bound the superposed information cost of the state $\xi = \rho_{NA}$. We are interested in the superposed states $\sigma^A_{NA},\sigma^B_{NA}$ of $\xi$ as defined in Section~\ref{Section:SuperposedDefs}. Recall that $\rho_{NA} = \sum_{xy} q_{xy} \altketbra{x} \otimes \altketbra{Y_{xy}^{NA}} \otimes \altketbra{y}$ for some $q_{xy}$. Let ${\spa{A}'} = \spa{O}_A \otimes \spa{M}_A \otimes \spa{R}_A \otimes A$ and ${\spa{B}'} = \spa{O}_B \otimes \spa{R}_B \otimes B$. We have $SIC(\xi) = I(X:M_B B'Y)_{\sigma^B_{NA}} + I(Y:XA')_{\sigma^A_{NA}}$.
Let also $\sigma^A_{2},\sigma^B_{2}$ the superposed states of $\rho_{2}$.

To proceed with the proof, we need the following lemmas and proposition.
\begin{lemma} \label{Lemma5} 
$
I(X:M_BB'Y)_{\sigma^B_{NA}} \le n\log(k) - H_{min}(X|M_BB'Y)_{\sigma^B_{2}} + t + 1
$
\end{lemma}
\begin{proof}
We have:
\begin{align*}
I(X:M_BB'Y)_{\sigma^B_{NA}} = H(X)_{\sigma^B_{NA}} - H(X|M_BB'Y)_{\sigma^B_{NA}} & \le n \log(k) - H(X|M_BB'Y)_{\sigma^B_{NA}} \\
& \le n\log(k) - H_{min}(X|M_BB'Y)_{\sigma^B_{NA}}.
\end{align*}

By definition, we have $ H_{min}(X|M_BB'Y)_{\sigma^B_{2}} = - \log (\Pr[\textrm{ Bob guesses } x \mid \textrm{Alice and Bob share } \sigma^B_{2} ])$. When Alice and Bob share $\sigma^B_2$, if Bob tries to determine whether he aborts or not, the state he shares with Alice conditioned on not aborting is $\sigma^B_{NA}$. Since Bob doesn't abort with probability greater than $2^{-t+1}$, we have 
\begin{align*}
\Pr[\textrm{Bob guesses } x  \mid \textrm{Alice and Bob share } \sigma^B_{2} ]) \ge 2^{-(t+1)} \Pr[\textrm{Bob guesses } x \mid  \textrm{Alice and Bob share } \sigma^B_{NA}]
\end{align*}
From there, we have
\begin{align*}
H_{min}(X|M_BB'Y)_{\sigma^B_{2}} & = - \log (\Pr[\textrm{ Bob guesses x } | \textrm{ Alice and Bob share } \sigma^B_{2} ]) \\
& \le - \log(\Pr[\textrm{ Bob guesses x }|  \textrm{ Alice and Bob share } \sigma^B_{NA}]) + t + 1\\
& = H_{min}(X|M_BB'Y)_{\sigma^B_{NA}} + t + 1.
\end{align*}
We conclude that
$I(X:M_BB'Y)_{\sigma^B_{NA}} \le n \log(k) - H_{min}(X|M_BB'Y)_{\sigma^B_{NA}} \le n\log(k) - H_{min}(X|M_BB'Y)_{\sigma^B_{2}} + t + 1$.
\qed \end{proof}

We now prove the following:
\begin{lemma}\label{Lemma6}
$H_{min}(X|M_BB'Y)_{\sigma_2^B} \le n\log(k) - m$.
\end{lemma}
\begin{proof}
 Let $\sigma'_{XM_BB'Y} = \Tr_{\spa{A}'}({\sigma_2^B})$, $\sigma'_{XB'Y} = \Tr_{\spa{A'M_B}} ({\sigma_2^B})$, and  $\sigma'_{B'Y} = \Tr_{\spa{XA'M_B}} ({\sigma_2^B})$. We have 
$H_{min}(X|M_BB'Y)_{\sigma_2^B} = H_{min}(X|M_BB'Y)_{\sigma'_{XM_BB'Y}}$.  First notice that 
 \begin{align}\label{Eq1}
 \sigma'_{XB'Y} = \frac{I_\spa{X}}{k^n} \otimes \sigma'_{B'Y}.
 \end{align}
 Moreover, we can write $\sigma'_{XMB'Y} = \sum_{M \in [m]} r_{M} \altketbra{M} \otimes \eta(M)_{\spa{XB'Y}}$ for some states $\{\eta(M)\}_{M}$ and $\sum_{M} r_M = 1$. Notice that $\sigma'_{XB'Y} = \sum_M r_M \eta(M)$. We have:
 \begin{align} \label{Eq2}
\sigma'_{XM_BB'Y} = \sum_{M \in [m]} r_{M} \altketbra{M} \otimes \eta(M)_{\spa{XB'Y}} \le I_{\spa{M}_B} \otimes \sigma'_{XB'Y}.
\end{align}
 
Using Equations~\ref{Eq1} and \ref{Eq2}, we have:
\begin{align*}
\sigma'_{XM_BB'Y} \le I_{\spa{M}_B} \otimes \sigma'_{XB'Y} & \le \frac{1}{k^n} I_\spa{X} \otimes I_{{M}_B} \otimes \sigma'_{B'Y} \\
& \le \frac{2^{m}}{k^n} I_\spa{X} \otimes \left(\frac{I_{{M}_B}}{2^{s_m}} \otimes \sigma'_{B'Y}\right).
\end{align*}
By definition of $H_{min}$(Section~\ref{Section:InformationTheory}), this gives $H_{min}(X|M_BB'Y)_{\sigma'_{XM_BB'Y}} \le n\log(k) - m$.
\qed \end{proof}
We now put everything together and prove the following.
\begin{proposition}
$I(X:M_BB'Y)_{\sigma^B_{NA}} \le m + t + 1$.
\end{proposition}
\begin{proof}
Combining Lemma~\ref{Lemma5} and Lemma~\ref{Lemma6}, we have
\begin{align*}
I(X:M_BB'Y)_{\sigma^B_{NA}} \le n\log(k) - H_{min}(X|M_BB'Y)_{\sigma^B_{2}} + t + 1 \le m + t + 1.
\end{align*}
\qed \end{proof}

Now let's analyze $\sigma^A_{NA}$. Here, Alice does not receive any message from Bob hence \mbox{$I(Y:XA')_{\sigma^A_2} = 0$.} As in Lemma~\ref{Lemma5}, we can show that $I(Y:XA')_{\sigma^A_{NA}} \le I(Y:XA')_{\sigma^A_2} + t + 1 = t + 1$.

Putting this all together, we have:
\[
SIC(\xi) = I(Y:XA')_{\sigma^A_{NA}} + I(X:M_BB'Y)_{\sigma^B_{NA}}
\le m + 2t + 2.
\]
To conclude the proof, recall from Section~\ref{Section:CommunicationProtocol} that $m = \frac{32 \log(|I||O|)}{\eps} ((t+1) + |\log(\eps)| + 5)$. From there, we conclude that 
\[SIC(\xi) \le \frac{32 \log(|I||O|)}{\eps} ((t+1) + |\log(\eps)| + 5) + 2t + 2,\]
which concludes the proof. 
\qed \end{proof}
We showed that $\xi$ satsfies all the desired properties of Theorem~\ref{Theorem:UpperBound}.
\addtocounter{theorem}{+1}

\section{Proof of Theorem \ref{Theorem:LowerBound}}\label{Appendix:LowerBound}
We now give a lower bound complementary to the upper bound described in Theorem~\ref{Theorem:UpperBound}.
\begin{theorem}\label{Theorem:LowerBound}
Consider a game $G = (I = [k],O,V,\mbox{Unif.})$ such that $\omega^*(G) = 1 - \eps$ and $\omega^*(G^n) = 2^{-t}$ with $t = o(n\eps)$. Let also $G^n_{1-\eps/32} = (I^n = [k^n],O^n,V',\mbox{Unif.})$ as defined in Section~\ref{Section:DefinitionMajorityGame}. For any game $G'=(I' = [k^n],O',V',p)$ and any state $
\rho = \sum_{x,y \in [k^n]} p_{xy} \altketbra{x}_{\spa{X}} \otimes \altketbra{\phi_{xy}} \otimes \altketbra{y}_\spa{Y},
$ satisfying
\begin{enumerate}
\item $H(XY)_\rho \ge 2n\log(k) - t$
\item $\omega^*(G'|\rho) \ge 1 - \eps/32$
\end{enumerate}
we have 
$SIC(\rho) \ge \Omega(n  \eps)$.
\end{theorem}
\begin{proof}
Fix any state $\rho$ of the form
\begin{align*}
\rho = \sum_{x,y \in [k^n]} p_{xy} \altketbra{x}_{\spa{X}} \otimes \altketbra{\phi_{xy}} \otimes \altketbra{y}_\spa{Y},
\end{align*}
satisfying properties $1.$ and $2.$ above. Property $2$ tells us that there is strategy that allows Alice and Bob to win $G'$ with high probability. We make them perform this strategy.

We first show that there is a large number of indices $i$ such that Alice and Bob win game $i$ with high probability with this stratagy and $H(\spa{X}_i,\spa{Y}_i)_{\rho}$ is large.

\begin{lemma}
Let $p_i = \Pr[\textrm{Alice and Bob win game } i \textrm{ using } \rho]$. Let $K = \{i : p_i \ge 1 - \eps/4\}$. Let $L = \{i : H(X_i,Y_i)_\rho \ge 2\log(k) - \frac{4t}{n} \}$. We have
\begin{align*}
|K| \ge 3n/4, \qquad |L| \ge 3n/4, \qquad \mbox{which implies} \quad |K \cap L| \ge n/2.
\end{align*}
\end{lemma}
\begin{proof}
$\frac{1}{n} \sum_{i \in [n]} p_i$ corresponds to the average number of games won by Alice and Bob. They win $G'$ if they win at least $(1 - \eps/32)$ games. Since they can win $G'$ with probability at least $1 - \eps/32$, we know that $\frac{1}{n} \sum_{i \in [n]} p_i \ge (1 - \eps/32)(1 - \eps/32) \ge 1 - \eps/16$. We have:
\begin{align*}
\sum_i p_i = \sum_{i \in K} p_i + \sum_{i \notin K} p_i \le |K| + (n - |K|)(1 - \eps/4) = n - (n - |K|)\eps/4,
\end{align*}
since $\sum_i p_i \ge n(1 - \eps/16)$, we have $n - (n - |K|)\eps/4 \ge n(1 - \eps/16)$ and $|K| \ge \frac{3n}{4}$.

Similarly,
 we have:
\begin{align*}
\sum_i H(X_iY_i)_{\rho} & = \sum_{i \in L} H(X_iY_i)_{\rho} + \sum_{i \notin L} H(X_iY_i)_{\rho} \\
& \le 2 |L| \log(k) + (n - |L|) (2\log(k) - \frac{4t}{n}) \\
& = 2n \log(k) - (n - |L|) \frac{4t}{n}.
\end{align*}
Since $\sum_i H(X_iY_i)_{\rho} \ge H(XY)_{\rho} = 2n\log(k) - t$, we have $2n\log(k) - (n - |L|) \frac{4t}{n} \ge 2n\log(k) - t$ which gives $|L| \ge \frac{3n}{4}$.

Putting this together, we have $|K \cap L| = |K| + |L| - |K \bigcup L| \ge |K| + |L| - n \ge n/2$.
\qed \end{proof}

The final step of the proof will be very similar to the proof of Proposition~$\ref{Prop2}$.

We start with a few notations. For a string $x = x_1,\dots,x_n \in [k^n]$, let $x_{-i}$ be the string in $[k^{n-1}]$ where we remove $x_i$ from $x$.
As in Section~\ref{Section:SuperposedDefs}, we define $\ket{L^A_y},\ket{L^B_x},\sigma^A,\sigma^B$. Also, let 
\begin{align*}
p_{x\cdot} = \sum_{y \in [k^n]} p_{xy} \ \ ; \ \ p_{\cdot y} = \sum_{x \in [k^n]} p_{xy}
\end{align*} and
\begin{align*}
p^i_{x_i,y_i} = \sum_{x',y': x'_i = x_i, y'_i = y_i} p_{x'y'} \ \ ; \ \
p^{-i}_{x_{-i},y_{-i}} = \sum_{x',y': x'_{-i} = x_{-i}, y'_{-i} = y_{-i}} p_{x'y'}.
\end{align*}

For each $i$, we rewrite $\rho$ as:
\begin{align*}
\rho = \sum_{x,y \in [k^n]} p_{xy} \altketbra{x_i}_{\spa{X}_i} \otimes \altketbra{x_{-i}}_{\spa{X}_{-i}} \otimes \altketbra{\phi_{xy}}_{AB} \otimes \altketbra{y_{-i}}_{\spa{Y}_{-i}} \otimes \altketbra{y_i}_{\spa{Y}_i}.
\end{align*}
We define  
\[
\ket{Z^i_{x_i,y_i}} = \sum_{x',y' \in [k^{n}] : x'_i = x_i, y'_i = y_i} \sqrt{p^{-i}_{x'_i,y'_i}} \ket{x'_{-i}}_{\spa{X}_{-i}} \otimes \ket{\phi_{x'y'}} \otimes \ket{y'_{-i}}_{\spa{Y}_{-i}}.
\]
Now, let $\gamma_i = \sum_{x_i,y_i \in [k]} p^i_{x_i,y_i} \altketbra{x_i} \otimes \altketbra{Z^i_{x_i,y_i}} \otimes \altketbra{y_i}$.
The state $\gamma_i$ corresponds to $\rho$ where the inputs in registers $\spa{X}_{-i},\spa{Y}_{-i}$ are in coherent superposition. In particular, Alice and Bob can go from $\gamma_i$ to $\rho$ by measuring the registers $\spa{X}_{-i}$ and $\spa{Y}_{-i}$ in the computational basis.

Using $\rho$, Alice and Bob can win the $i^{th}$ instance of $G$ with probability $p_i$. This means that they can win this $i^{th}$ instance of $G$ when sharing $\gamma_i$ with probability at least $p_i$.

Now, consider $\sigma^A_i,\sigma^B_i$ the 2 superposed states of $\gamma_i$ as defined in Section~\ref{Section:SuperposedDefs}. We first show the following:
\begin{lemma}
If $t \le \frac{c_0 \eps n}{4}$ then $\forall i \in K \cap L$, $I(Y_i:XA)_{\sigma^A_i} + I(X_i:BY)_{\sigma^B_i} = \Omega(\eps)$.
\end{lemma}
\begin{proof}
Consider $i \in K \cap L$. Since $i \in L$, we have $H(X_iY_i)_{\gamma_i} \ge 2\log(k) - 4t/n \ge 2\log(k) - c_0\eps$. Using Claim~\ref{Vadhan}, we have $\Delta(p^i,\mbox{Unif.}) \le c_0\eps$ or in other words that $\frac{1}{2} \sum_{x_i,y_i \in [k]} |p^i_{x_i,y_i} - \frac{1}{k^2}| \le c_0\eps$. Since $i \in K$, we have $\omega^*(G'_i|\gamma_i) \ge 1 - \eps/4$ for $G'_i = (I,O,V,p^i)$. Using Theorem~\ref{Theorem:OneShot}, we conclude that $SIC(\gamma_i) = I(Y_i:XA)_{\sigma^A_i} + I(X_i:BY)_{\sigma^B_i} = \Omega(\eps)$.
\qed \end{proof}
We can now finish the proof. The above lemma holds for our $t$ since $t = o(\eps n)$. First notice that $Tr_{\spa{Y}_{-i}}(\sigma^A_i) = Tr_{\spa{Y}_{-i}} (\sigma^A)$ hence $I(Y_i:XA)_{\sigma^A_i} = I(Y_i:XA)_{\sigma^A}$. Similarly, we have  $I(X_i:BY)_{\sigma^B_i} = I(X_i:BY)_{\sigma^B}$ which gives 
\begin{align*}
I(Y_i:XA)_{\sigma^A_i} + I(X_i:BY)_{\sigma^B_i} = I(X_i:BY)_{\sigma^A} + I(Y_i:XA)_{\sigma^B}
\end{align*}
and hence  
\begin{align*}
\forall i \ \in K \cap L, \ I(X_i:BY)_{\sigma^A} + I(Y_i:XA)_{\sigma^B} = \Omega(\eps).
\end{align*}
To conclude, firs we write
\begin{align*}
SIC(\rho) & = I(Y:XA)_{\sigma^A} + I(X:BY)_{\sigma^B} \\
& = H(Y)_{\sigma^A} - H(Y|XA)_{\sigma^A} + H(X)_{\sigma^B} - H(X|BY)_{\sigma^B} \\
& \ge H(XY)_{\rho} - H(Y|XA)_{\sigma^A} - H(X|BY)_{\sigma^B} \\
& \ge 2n\log(k) - t - H(Y|XA)_{\sigma^A} - H(X|BY)_{\sigma^B} \\
& \ge \sum_{i \in [n]} H(Y_i)_{\sigma^A} - H(Y|XA)_{\sigma^A} + \sum_{i \in [n]} H(X_i)_{\sigma^B} - H(X|BY)_{\sigma^B} - t  \\
& \ge \sum_{i \in [n]} H(Y_i)_{\sigma^A} - H(Y_i|XA)_{\sigma^A} + \sum_{i \in [n]} H(X_i)_{\sigma^B} - H(X_i|BY)_{\sigma^B} - t \\
& = \sum_{i \in [n]} I(Y_i:XA)_{\sigma^A} + I(X_i:BY)_{\sigma^B} - t \\
& \ge \sum_{i \in K \cap L} I(Y_i:XA)_{\sigma^A} + I(X_i:BY)_{\sigma^B} -t \\
& = \Omega(n\eps) - t = \Omega(n\eps) \qquad  \mbox{since}  \ |K \cap L| \ge n/2 \textrm{ and } t = o(n\eps).
\end{align*}
\qed \end{proof}

\section{Games with complete support} \label{Appendix:GeneralGames}

In this Appendix we prove Corollary \ref{corollary:Final}.
The idea is the following. Starting from any game $G$ with complete support, we define a new game $H$ that can be interpreted as follows:
\begin{itemize}
\item With some probability Alice and Bob play $G_U$, a variant of $G$ on the uniform distribution
\item If they are not in the previous case, they win no matter what they answer
\item They know in which case they are thanks to an extra input bit
\item If they ignore the extra bit of information, they play the original game.
\end{itemize}
In Lemma \ref{LemmaCompleteSupp1} we prove that $H^n$ has a larger value than $G^n$, which intuitively follows from the fact that players can just ignore the extra bits. Since the difficulty of winning $H^n$ comes from the indices where Alice and Bob must play $G_U$, in Lemma \ref{LemmaCompleteSupp2} we show that the winning probability of $H^n$ is bounded by the winning probability of a parallel repetition of $G_U$. We have an exponential decay because $G_U$ meets the requirements of Theorem \ref{Theorem:Final}. To finish, we relate $\omega^*(G)$ to $\omega^*(G_U)$ in Lemma \ref{LemmaCompleteSupp3} and we prove Corollary \ref{corollary:Final}.

\

Let us start with some definitions. 
Let $G = (I,O,V,p)$ with $|I| = k$. Let $\alpha_{min} = \min_{xy} \{k^2 p_{xy} \}$ and $\alpha_{max} = \max_{xy} \{k^2 p_{xy} \}$. 
 We have:
\begin{align*}
\forall (x,y) \in I^2 \quad  \frac{\alpha_{min}}{k^2} \le p_{xy} \le \frac{\alpha_{max}}{k^2}.
\end{align*}
Let $U$ be the uniform distribution on $I^2$. By seeing $p$ and $U$ as vectors indexed by $(x,y)$, we can rewrite the above as $\alpha_{min} U \le p \le \alpha_{max} U$.
Let $p'$ the probability distribution satisfying $p = \alpha_{min} U + (1 - \alpha_{min}) p'$. Let $\tilde I = \zo \times I$ and $q$ be the distribution on $\tilde I^2$ such that $q_{0x0y} = \frac{\alpha_{min}}{k^2}$, $q_{1x1y} = (1 - \alpha_{min})p'_{xy}$ and $q_{0x1y} = q_{1x0y} = 0$. We have:
$$
q_{0x0y} + q_{1x1y} = \frac{\alpha_{min}}{k^2} + (1 - \alpha_{min})p'_{xy} = p_{xy}.
$$
We define the game $H = (\tilde I,O,\tilde V,q)$ with the following winning predicate:
\begin{itemize}
\item $\tilde V(ab | 0x0y) = V(ab|xy)$ for all $a,b \in O^2$.
\item $\tilde V(ab|1x1y) = 1$ for all $a,b\in O^2$.
\end{itemize}
This means that if Alice and Bob's extra bit is $0$ the predicate is the same than the predicate of $G$, while if the extra bit is $1$ they always win. 
Notice that for each $c \in \zo$, we have $\tilde V(ab | cxcy) \ge V(ab|xy)$. 
Now consider the parallel repetition. Let $\tilde x,\tilde y \in \tilde I^n$, where we write $\tilde x_i = c_i x_i$ and $\tilde y_i = c_i y_i$ with $x_i,y_i \in I$ and $c_i$ being the extra bit. Let $\tilde V'$ and $V'$ be the predicates for $H^n$ and $G^n$, respectively. Then for all $a,b$, we have 
\begin{equation} \label{Eq:parallelVtilde}
\tilde V'(ab|\tilde x\tilde y) = \Pi_i (\tilde V(a_ib_i | \tilde x_i \tilde y_i)) \ge \Pi_i (V(a_ib_i | xy)) = V'(ab|xy).
\end{equation}

\begin{lemma} \label{LemmaCompleteSupp1}
$\omega^*(G^n) \le \omega^*(H^n)$.
\end{lemma}
\begin{proof}
Fix an optimal strategy for $G^n$. Let $P(ab|xy)$ the probability that Alice and Bob output $a,b\in O^{n}$ on inputs $x,y\in I^{n}$ when applying such strategy for $G^n$. We have
$$
\omega^*(G^n) = \sum_{xyab} p_{xy} P(ab|xy) V'(ab|xy). $$
Define the following strategy for $H^n$. Alice and Bob, on inputs $\tilde x, \tilde y \in \tilde I^n$ according to $q^n$,
ignore the extra bit and apply the above optimal strategy for $G^n$ on inputs $x,y$. 
 Let $\tilde P(ab|\tilde x \tilde y)$ the probability of outputs $a,b$ on inputs $\tilde x \tilde y$ with this strategy. We have $\tilde P(ab|\tilde x \tilde y) = P(ab|xy)$.
Also note that for all $x,y \in I^n$ and $c \in \zo^n$ we can write $q_{cxcy} = p_{xy} r_{xyc}$ where 
$r_{xyc} \ge 0 $ and $\sum_{c} r_{xyc} = 1.$
It follows from above and \eqref{Eq:parallelVtilde} that
\begin{align*}
\omega^*(H^n) & \ge \sum_{xyab} \sum_c q_{cxcy} \tilde P(ab|\tilde xyc) V'^n(ab|cxyc)  \\
& \ge \sum_{xyab} p_{xy} \sum_c r_{xyc} P(ab|xy) V'^n(ab|cxyc) \\
& \ge \sum_{xyab} p_{xy} \sum_c r_{xyc} P(ab|xy) V^n(ab|xy) \\
& = \sum_{xyab} p_{xy} P(ab|xy) V^n(ab|xy) = \omega^*(G^n).
\end{align*}
\qed \end{proof}

We now want to upper bound $\omega^*(H^n)$. Let $G_U$ the game $G$ on the uniform distribution and let $\eps_U = 1 - \omega^*(G_U)$. We prove the following.
\begin{lemma} \label{LemmaCompleteSupp2}
$
 \omega^*(H^n)  \le (1 - \eps^2_U)^{\Omega(\frac{n \alpha_{min}}{\log(|I||O|)} - \alpha_{min} |\log(\eps)|)}.
$
\end{lemma}
\begin{proof}
Thanks to the extra bit, we can interpret $H$ as follows:
\begin{itemize}
\item With probability $\alpha_{min}$, Alice and Bob play $G_U$
\item With probability $1 - \alpha_{min}$, they win on any output
\item They know in which case they are
\end{itemize}
If there are $k$ instances of $G_U$, Alice and Bob win the whole game if and only if they win these $k$ instances of $G_U$.
The probability that $i$ instances of $G_U$ occur is equal to ${n \choose i} \alpha_{min}^i (1 - \alpha_{min})^{n-i}$. This gives
$$
\omega^*(H^n) \le \sum_{i = 0}^n {n \choose i} \alpha_{min}^i (1 - \alpha_{min})^{n-i} \omega^*(G_U^{i}).
$$
Since $G_U$ is on the uniform distribution, we have that
$\omega^*(G_U^i) \le (1 - \eps^2_U)^{\Omega(\frac{i}{\log(|I||O|)} - |\log(\eps)|)}$ by Theorem \ref{Theorem:Final}.
This gives 
$$
\omega^*(H^n) \le \sum_{i = 0}^n {n \choose i}  \alpha_{min}^i (1 - \alpha_{min})^{n-i} (1 - \eps^2_U)^{\Omega(\frac{i}{\log(|I||O|)} - |\log(\eps)|)}.
$$
We can show by analytic calculations that 
$$
 \sum_{i = 0}^n {n \choose i}  \alpha_{min}^i (1 - \alpha_{min})^{n-k} Y^i = (\alpha_{min}(Y-1 ) +1)^n.
 $$
By plugging this in the above inequality, we obtain
 \begin{align*}
 \omega^*(H^n) & \le (1 - \eps^2_U)^{\Omega(\frac{n \alpha_{min}}{\log(|I||O|)} - \alpha_{min} |\log(\eps)|)}.
 \end{align*}
 \qed \end{proof}
 
 Now we relate the values of $G$ and $G_U$. Let $\omega^*(G) = 1 - \eps$. 
 \begin{lemma} \label{LemmaCompleteSupp3}
$\eps \le \alpha_{max} \eps_U$.
\end{lemma}
\begin{proof}
Fix an optimal strategy for $G_U$ and let $P(ab|xy)$ the probability of outputs $a,b$ on inputs $x,y$ for this strategy. We have 
\begin{align*}
\eps_U & = \sum_{xyab} \frac{1}{k^2} P(ab|xy) (1 - V(ab|xy)) \\
&\ge \sum_{xyab}  \frac{p_{xy}}{\alpha_{max}}  P(ab|xy) (1 - V(ab|xy)) \\
& \ge \frac{\eps}{\alpha_{max}}.
\end{align*}
\qed \end{proof}

Finally, we combine all of the above in the final corollary

\addtocounter{corollary}{-4}
\begin{corollary}
Let $G=(I,O,V,p)$ be a game with complete support and $\omega^*(G) = (1 - \eps)$. Then,
$$\omega^*(G^n) \le (1 - \eps^2)^{\Omega(\frac{n}{Q\log(|I||O|)} - \frac{|\log(\eps)|}{Q})},$$
where $Q= \frac{k^2 \max_{xy} p_{xy}^2 }{\min_{xy} p_{xy} }$.
\end{corollary}
\begin{proof}
By chaining the previous lemmas, we obtain
\begin{align*}
\omega^*(G^n) \le \omega^*(H^n) & \le (1 - \eps^2_U)^{\Omega(\frac{n \alpha_{min}}{\log(|I||O|)} - \alpha_{min} |\log(\eps)|)} \\
& \le (1 - \eps^2)^{\Omega(\frac{n \alpha_{min}}{\alpha^2_{max} \log(|I||O|)} - \frac{\alpha_{min} |\log(\eps)|)}{\alpha^2_{max}}}
 \le (1 - \eps^2)^{\Omega(\frac{n}{Q\log(|I||O|)} - \frac{|\log(\eps)|}{Q})}.
\end{align*} 
\qed \end{proof}

\end{appendix}

\end{document}